\documentclass{theoretics}

\newcommand\numberthis{\addtocounter{equation}{1}\tag{\theequation}}
\usepackage{bm}
\usepackage{mathtools}
\DeclarePairedDelimiter\ceil{\lceil}{\rceil}
\DeclarePairedDelimiter\floor{\lfloor}{\rfloor}

\usepackage{thm-restate}

\title{A Robust Version of Heged\H{u}s's Lemma, with Applications}

\ThCSauthor{Srikanth Srinivasan}{srikanth@cs.au.dk}[https://orcid.org/0000-0001-6491-124X]
\ThCSaffil{Department of Computer Science, Aarhus University, Denmark}
\ThCSthanks{Work done while at the Department of Mathematics, Indian Institute of Technology Bombay, Mumbai, India.  Supported by MATRICS grant MTR/2017/000958 awarded by SERB, Government of India. A preliminary version
of this paper appeared at STOC 2020.}

\ThCSyear{2023}
\ThCSarticlenum{5}
\ThCSreceived{Feb 28, 2022}
\ThCSrevised{Oct 21, 2022}
\ThCSaccepted{Dec 20, 2022}
\ThCSpublished{Feb 24, 2023}
\ThCSkeywords{Polynomial approximation, Boolean functions, Probabilistic degree, Coin Problem}
\ThCSdoi{10.46298/theoretics.23.5}
\ThCSshortnames{S. Srinivasan} 
\ThCSshorttitle{A Robust Version of Heged\H{u}s's Lemma}

\newcommand{\prob}[2]{\mathop{\mathrm{Pr}}_{#1}\left[#2\right]}
\newcommand{\avg}[2]{\mathop{\textbf{E}}_{#1}[#2]}

\newcommand{\F}{\mathbb{F}}

\newcommand{\AC}{\mathrm{AC}}

\newcommand{\ip}[2]{\langle #1, #2 \rangle}

\newcommand{\mc}[1]{\mathcal{#1}}

\newcommand{\pdeg}{\mathrm{pdeg}}
\newcommand{\spec}{\mathrm{Spec}\,}
\newcommand{\per}{\mathrm{per}}
\newcommand{\sB}{s\mathcal{B}}
\newcommand{\Maj}{\mathrm{Maj}}
\newcommand{\Thr}{\mathrm{Thr}}
\newcommand{\MOD}{\mathrm{MOD}}

\newcommand{\EThr}{\mathrm{EThr}}

\newcommand{\mb}{\mathbb}
\newcommand{\tx}{\text}

\newcommand{\charac}{\tx{char}}

\begin{document}
	
	\maketitle
	
	\begin{abstract}
		Heged\H{u}s's lemma is the following combinatorial statement regarding polynomials over finite fields. Over a field $\F$ of characteristic $p > 0$ and for $q$ a  power of $p$, the lemma says that any multilinear polynomial $P\in \F[x_1,\ldots,x_n]$ of degree less than $q$ that vanishes at all points in $\{0,1\}^n$ of some fixed Hamming weight $k\in [q,n-q]$ must also vanish at all points in $\{0,1\}^n$ of weight $k + q$. This lemma was used by Heged\H{u}s (2009) to give a solution to \emph{Galvin's problem}, an extremal problem about set systems; by Alon, Kumar and Volk (2018) to improve the best-known multilinear circuit lower bounds; and by Hrube\v{s}, Ramamoorthy, Rao and Yehudayoff (2019) to prove optimal lower bounds against depth-$2$ threshold circuits for computing some symmetric functions. 
		
		In this paper, we formulate a robust version of Heged\H{u}s's lemma. Informally, this version says that if a polynomial of degree $o(q)$ vanishes at most points of weight $k$, then it vanishes at many points of weight $k+q$. We prove this lemma and give the following three different applications.
		\begin{itemize}
		\item Degree lower bounds for the coin problem: The \emph{$\delta$-Coin Problem} is the problem of distinguishing between a coin that is heads with probability $((1/2) + \delta)$ and a coin that is heads with probability $1/2$. We show that over a field of positive (fixed) characteristic, any polynomial that solves the $\delta$-coin problem with error $\varepsilon$ must have degree $\Omega(\frac{1}{\delta}\log(1/\varepsilon)),$ which is tight up to constant factors.
		
		\item Probabilistic degree lower bounds: The \emph{Probabilistic degree} of a Boolean function is the minimum $d$ such that there is a random polynomial of degree $d$ that agrees with the function at each point with high probability. We give tight lower bounds on the probabilistic degree of \emph{every} symmetric Boolean function over positive (fixed) characteristic. As far as we know, this was not known even for some very simple functions such as unweighted Exact Threshold functions, and constant error.
		
		\item A robust version of the combinatorial result of Heged\H{u}s (2009) mentioned above.
		\end{itemize}
	\end{abstract}
	
	\section{Introduction}
	
	The Polynomial Method is a technique of great utility in both Theoretical Computer Science and Combinatorics. The idea of associating polynomials with various combinatorial objects and then using algebraic or geometric techniques to analyze them has proven useful in many settings including, but not limited to, Computational Complexity (Circuit lower bounds~\cite{Razborov,Smolensky87,Beigel-survey,Wil_ACC}, Pseudorandom generators~\cite{Braverman}), Algorithm design (Learning Algorithms~\cite{LMN,KlivansServedio,KlivansODonnellServedio04}, Satisfiability algorithms~\cite{Wil_ACC,WilACCthr}, Combinatorial algorithms~\cite{Will-APSP,AbboudWY,AW}), and Extremal Combinatorics~\cite{Guthbook,CLP,EG}.
	
	The engine that drives the proofs of many of these results is our understanding of combinatorial and algebraic properties of polynomials. In this paper, we investigate another such naturally stated property of polynomials defined over the Boolean cube $\{0,1\}^n$ and strengthen known results in this direction. We then apply this result to sharpen known results in theoretical computer science and combinatorics. 
	
	The question we address is related to how well low-degree polynomials can `distinguish' between different layers of the Boolean cube $\{0,1\}^n.$ For $m\in \{0,\ldots,n\}$,  let $\{0,1\}^n_m$ be the elements of $\{0,1\}^n$ of Hamming weight exactly $m$. As a first approximation, let us say that a polynomial $P\in \F[x_1,\ldots,x_n]$ (here $\F$ is some field) distinguishes between level sets $\{0,1\}^n_k$ and $\{0,1\}^n_K$ if it vanishes at all points in the former set and at no point of the latter. Note that the ability of low-degree polynomials to do this depends on the properties of the underlying field $\F$: when $\F = \mathbb{Q}$ (or any field of characteristic $0$), the simple polynomial $\left(\sum_{i=1}^n x_i\right) - k$ does the job. However, if the field $\F$ has positive characteristic $p$ and more specifically if $K-k$ is divisible by~$p$, then this simple polynomial no longer works and the answer is not so clear.
	
	In this setting, a classical theorem of Lucas tells us that if $q$ is the largest power of $p$ dividing $K-k$, then there is a polynomial of degree $q$ that distinguishes between $\{0,1\}^n_k$ and $\{0,1\}^n_K.$ A very interesting lemma of Heged\H{u}s~\cite{Hegedus} shows that this is tight even if we only require $P$ to be non-zero at \emph{some} point of $\{0,1\}^n_K.$
	More precisely, Heged\H{u}s's lemma shows the following.\footnote{The lemma is usually stated~\cite{Hegedus,AKV,HRRY} for a more restricted choice of parameters. However, the known proofs extend to yield the stronger statement given here. A proof of a more general statement can be found in~\cite[Theorem 1.5]{SV-arxiv}.} 
	
	\begin{lemma}[Heged\H{u}s's lemma]
	\label{lem:Hegedus}
	Let $\mathbb{F}$ be a field of characteristic $p > 0$. Fix any positive integers $n,k,q$ such that $k\in [q,n-q],$ and $q$ a power of $p$. If $P\in \F[x_1,\ldots,x_n]$ is any  polynomial that vanishes at \emph{all} $a\in \{0,1\}^n_k$ but does not vanish at \emph{some} $b\in \{0,1\}^n_{k+q}$, then  $\deg(P) \geq q.$
	\end{lemma}
	
	This lemma was first proved in~\cite{Hegedus} using Gr\"{o}bner basis techniques. An elementary proof of this was recently given by the author  and independently by Alon~(see~\cite{HRRY}) using the Combinatorial Nullstellensatz.
	
	Heged\H{u}s's lemma has been used to resolve various questions in both combinatorics and theoretical computer science.
	
	\begin{itemize}
	\item Heged\H{u}s used this lemma to give an alternate solution to a problem of Galvin, which is stated as follows. Given a positive integer $n$ divisible by $4$, what is the smallest size $m = m(n)$ of a family $\mc{F}$ of $(n/2)$-sized subsets of $[n]$ such that for any $S\subseteq [n]$ of size $n/2$, there is a $T\in \mc{F}$ with $|T\cap S| = n/4$? It is easy to see that $m(n)\leq n/2$ for any $n$. A matching lower bound was given by Enomoto, Frankl, Ito and Nomora~\cite{EFIN} in the case that $t := (n/4)$ is odd. Heged\H{u}s used the above lemma to give an alternate proof of a lower bound of $n$ in the case that $t$ is an odd prime. His proof was subsequently strengthened to a linear lower bound for all $t$ by Alon et al.~\cite{AKV} and more recently to a near-tight lower bound of $(n/2)-o(n)$ for all $t$ by Hrube\v{s} et al.~\cite{HRRY}. Both these results used the lemma above.
	
	\item Alon et al.~\cite{AKV} also used Heged\H{u}s's lemma to prove bounds for generalizations of Galvin's problem. Using this, they were able to prove improved lower bounds against \emph{syntatically multilinear algebraic circuits.} These are algebraic circuits that compute multilinear polynomials in a ``transparently multilinear'' way (see e.g.~\cite{SY} for more). Alon et al. used Heged\H{u}s's lemma to prove near-quadratic lower bounds against syntactically multilinear algebraic circuits computing certain explicitly defined multilinear polynomials, improving on an earlier $\tilde{\Omega}(n^{4/3})$ lower bound of Raz, Shpilka and Yehudayoff~\cite{RSY}.
	
	\item Hrube\v{s} et al.~\cite{HRRY} also used Heged\H{u}s's lemma to answer the following question of Kulikov and Podolskii~\cite{KP} on depth-$2$ threshold circuits. What is the smallest $k = k(n)$ such that there is a depth-$2$ circuit made up of Majority\footnote{The Majority function is the Boolean function $f$ which accepts  exactly those inputs that have more $1$s than $0$s.} gates of fan-in at most $k$ that computes the Majority function on $n$ bits? Using Heged\H{u}s's lemma, Hrube\v{s} et al. showed an asymptotically tight lower bound of $n/2-o(n)$ on $k(n)$. 
	\end{itemize}
	
	\paragraph{Main Result.} Our main result in this paper is a `robust' strengthening of Heged\H{u}s's lemma. Proving `robust' or `stability' versions of known results is standard research direction in combinatorics. Such questions are usually drawn from the following template. Given the fact that objects that satisfy a certain property have some fixed structure, we ask if a similar structure is shared by objects that `almost' or `somewhat' satisfy the property. 
	
	In our setting, we ask if we can recover the degree lower bound in Heged\H{u}s's lemma even if we have a polynomial $P$ that `approximately' distinguishes between $\{0,1\}^n_k$ and $\{0,1\}^n_{k+q}$: this means that the polynomial $P$ vanishes at `most' points of weight $k$ but is non-zero at `many' points of weight $k+q$. Our main lemma is that under suitable definitions of `most' and `many', we can recover (up to constant factors) the same degree lower bound as in Lemma~\ref{lem:Hegedus} above.
	
	\begin{lemma}[Main Result (Informal)]
	\label{lem:main-informal}
	Assume that $\F$ is a field of characteristic $p$. Let $n$ be a growing parameter and assume  we have positive integer parameters $k,q$ such that $100q < k < n-100q$ and $q$ is a power of $p$. For $\varepsilon = \varepsilon(n,k,q),$ if $P\in\F[x_1,\ldots,x_n]$ that vanishes at a $(1-\varepsilon)$-fraction  of points of $\{0,1\}^n_k$ but does not vanish at an $\varepsilon^{0.0001}$ fraction of points of $\{0,1\}^n_{k+q},$ then $\deg(P) = \Omega(q).$
	\end{lemma}
	
	\begin{remark}
	\label{rem:intromainlem}
	\begin{enumerate}
	\item To keep the exposition informal, we have not specified exactly what $\varepsilon$ is in the above lemma. However, we note below that the $\varepsilon$ chosen is nearly the best possible in the sense that if $\varepsilon$ is appreciably increased, then there is a sampling-based construction of a polynomial~$P$ of degree $o(q)$ satisfying the hypothesis of the above lemma (see Section~\ref{sec:tightness}).
	
	\item	The reader might wonder why the lemma above is a strengthening of Heged\H{u}s's lemma, given that we require the polynomial $P$ to be non-zero at many points of weight $k+q$, which is a seemingly stronger condition than required in Lemma~\ref{lem:Hegedus}. However, this is in fact a weaker condition. This is because of the following simple algebraic fact: if there is a polynomial $P$ of degree at most $d$ satisfying the hypothesis of Lemma~\ref{lem:Hegedus} (i.e. vanishing at all points of weight~$k$ but not at some point of weight $k+q$), then there is also a polynomial~$Q$ of degree at most~$d$ that vanishes at all points of weight~$k$ but does not vanish at a \emph{significant fraction} (at least a $(1-1/p)$ fraction) of points of weight $k+q$. We give a short proof of this in Appendix~\ref{appsec:lem2vslem1}. Hence, the above lemma is indeed a generalization of Lemma~\ref{lem:Hegedus} (up to the constant-factor losses in the degree lower bound). 
	\end{enumerate}
	\end{remark}
	
	\paragraph{Applications.} Our investigations into robust versions of Heged\H{u}s's lemma were motivated by questions in computational complexity theory. Using our main result, we are able to sharpen and strengthen known results in complexity as well as combinatorics. 
	
	\begin{enumerate}
	\item \textbf{Degree bounds for the Coin Problem:} For a parameter $\delta \in [0,1/2],$ we define the \emph{$\delta$-coin problem} as follows. We are given $N$ independent tosses of a coin, which is promised to either be of bias $1/2$ (i.e. unbiased) or $(1/2 )- \delta$, and we are required to guess which of these is the case with a high degree of accuracy, say with error probability at most $\varepsilon$. (See Definition~\ref{def:coinproblem} for the formal definition.)
	
	The coin problem has been studied in a variety of settings in complexity theory (see, e.g.~\cite{ABO,Valiant,Viola,SV,BV,CGR}) and for various reasons such as understanding the power of randomness in bounded-depth circuits, the limitations of blackbox hardness amplification, and devising pseudorandom generators for bounded-width branching programs. More recently, Limaye et al.~\cite{LSSTV-SICOMP} proved optimal lower bounds on the size of $\AC^0[\oplus]$~\footnote{Recall that these are bounded-depth circuits made up of AND, OR and $\oplus$ gates.} circuits solving the $\delta$-coin problem with constant error, strengthening an earlier lower bound of Shaltiel and Viola~\cite{SV}. This led to the first class of explicit functions for which we have tight (up to polynomial factors) $\AC^0[\oplus]$ lower bounds. These bounds were in turn  used by Golovnev, Ilango, Impagliazzo, Kabanets, Kolokolova and Tal~\cite{GIIKKT} to resolve a long-standing open problem regarding the complexity of MCSP in the $\AC^0[\oplus]$ model, and by Potukuchi~\cite{Potukuchi} to prove lower bounds for Andreev's problem.
	
	A key result in the lower bound of Limaye et al.~\cite{LSSTV-SICOMP} was a tight lower bound on the degree of any polynomial $P\in \F[x_1,\ldots,x_N]$ that solves the $\delta$-coin problem with constant error: they showed that any such polynomial $P$ must have degree at least $\Omega(1/\delta).$ As noted by Agrawal~\cite{Agrawal}, this is essentially equivalent to a recent result of Chattopadhyay, Hatami, Lovett and Tal~\cite{CHLT} on the level-$1$ Fourier coefficients of low-degree polynomials over finite fields, which in turn is connected to an intriguing new approach~\cite{CHLT} toward constructing pseudorandom generators secure against $\AC^0[\oplus].$
	
	Using the robust Heged\H{u}s lemma, we are able to strengthen the degree lower bound of~\cite{LSSTV-SICOMP} to a tight degree lower bound for \emph{all errors}. Specifically, we show that over any field $\F$ of fixed positive characteristic $p$, any polynomial $P$ that solves the $\delta$-coin problem with error $\varepsilon$ must have degree $\Omega(\frac{1}{\delta}\log(1/\varepsilon))$, which is tight for all $\delta$ and $\varepsilon.$
	
	\item \textbf{Probabilistic degrees of symmetric functions:} In a landmark paper~\cite{Razborov}, Razborov showed how to use polynomial approximations to prove lower bounds against $\AC^0[\oplus]$. The notion of polynomial approximation introduced (implicitly) in his result goes by the name of \emph{probabilistic polynomials,} and is defined as follows. An $\varepsilon$-error probabilistic polynomial of degree $d$ for a Boolean function $f:\{0,1\}^n\rightarrow \{0,1\}$ is a random  polynomial~$\bm{P}$ of degree at most $d$ that agrees with $f$ at each point with probability at least $1-\varepsilon$. The $\varepsilon$-error probabilistic degree of $f$ is the least $d$ for which this holds. (Roughly speaking, a low-degree probabilistic polynomial for $f$ is an efficient randomized algorithm for $f$, where we think of polynomials as algorithms and degree as a measure of efficiency.)
	
	Many applications of polynomial approximation in complexity theory~\cite{Beigel-survey} and algorithm design~\cite{Will-FSTsurvey} use probabilistic polynomials and specifically bounds on the probabilistic degrees of various \emph{symmetric} Boolean functions.\footnote{Recall that a Boolean function $f:\{0,1\}^n\rightarrow \{0,1\}$ is said to be symmetric if its output depends only on the Hamming weight of its input.} Motivated by this,  in a recent result with Tripathi and Venkitesh~\cite{STV}, we gave a near-tight characterization on the probabilistic degree of every symmetric Boolean function. Unfortunately, however, our upper and lower bounds were separated by logarithmic factors. This can be crucial: in certain algorithmic applications (see, e.g.,~\cite[Footnote, Page 138]{AW}), the appearance or non-appearance of an additional logarithmic factor in the degree can be the difference between (say) a truly subquadratic running time of $N^{2-\varepsilon}$ and a running time of $N^{2-o(1)}$, which might be less interesting.
	
	In the case of characteristic $0$ (or growing with $n$), such gaps look hard to close since we don't even understand completely the probabilistic degree of simple functions like the OR function~\cite{MNV,HS,BHMS}. However, in positive (fixed) characteristic, there are no obvious barrriers. Yet, even in this case, the probabilistic degree of very simple symmetric Boolean functions like the \emph{Exact Threshold functions} (functions that accept inputs of exactly one Hamming weight) remained unresolved until this paper. 
	
	In this paper, we resolve this question and more. We are able to give a tight (up to constants) lower bound (matching the upper bounds in~\cite{STV}) on the probabilistic degree of \emph{every} symmetric function over fields of positive (fixed) characteristic. 
	
	\item \textbf{Robust version of Galvin's problem:} Given that Heged\H{u}s's lemma was used to solve Galvin's problem, it is only natural that we consider the question of using the robust version to solve a robust version of Galvin's problem. More precisely, we consider the minimum size $m = m(n,\varepsilon)$ to be the minimum size of a family $\mc{F}$ of $(n/2)$-sized subsets of $[n]$  such that for all but an $\varepsilon$-fraction of sets $S$ of size $n/2$, there is a set $T\in \mc{F}$ such that $|S\cap T|  = n/4.$ 
	
	Following the proof of Galvin's theorem from Heged\H{u}s's lemma, we can prove a lower bound of $\Omega(\sqrt{n\log(1/\varepsilon}))$ for the above version of Galvin's problem for any $\varepsilon\in [2^{-n},1/2].$ Note that this interpolates smoothly between a bound of $\Omega(\sqrt{n})$ for constant $\varepsilon$ and $\Omega(n)$ for $\varepsilon = 2^{-\Omega(n)}$, both of which are tight. For general $\varepsilon$ in between these two extremes, we do not know if our bounds are tight (we suspect they are). However, our bounds \emph{are} tight for every $\varepsilon$ for a natural generalization of the above problem, where we allow intersections of any size (and not just $n/4$). We refer the reader to Section~\ref{sec:galvin} for details.
	\end{enumerate}
	
	\paragraph{Proof Outline.} We observe that the main lemma (Lemma~\ref{lem:main-informal}) is quite similar to classical polynomial approximation results of Razborov~\cite{Razborov} and Smolensky~\cite{Smolensky87, Smolensky93} (see also~\cite{Szegedy-thesis}). The main difference is that while these results hold for polynomials approximating some function on the whole cube $\{0,1\}^n$, the lemma deals with polynomial approximations that are more `local' in that they are restricted on just two layers of the cube. Nevertheless, we can show that the basic proof strategy of Smolensky (or more specifically a variant as in~\cite{ABFR94,KS}) can be used to prove our lemma as well.
	
	The main point of difference from these standard proofs is the employment of a result from discrete geometry due to Nie and Wang~\cite{NW}, that allows us to bound the size of the \emph{closure}\footnote{The \emph{degree-$D$ closure} $\mathrm{cl}_D(E)$ of a set $E$ is the set of points where any degree-$D$ polynomial $Q$ vanishing throughout $E$ is forced to vanish.} of a small set of points in the cube. This is a well-studied object in coding theory~\cite{Wei} and combinatorics~\cite{CL, KeevashSudakov, NW}, and turns out to be a crucial ingredient in our proof.
	
	For the application to the coin problem, we show that if a polynomial $P$ solves the coin problem (see Definition~\ref{def:coinproblem} for the formal definition of this), then it can be used to distinguish between Hamming weights $k$ and $k+q$ for $k$ and $q$ as in Lemma~\ref{lem:main-informal}. This reduction is done by a simple sampling argument. The degree lower bound in Lemma~\ref{lem:main-informal} then implies the desired degree lower bound on the degree of $P$.
	
	In the other applications to probabilistic degree and the robust version of Galvin's problem, the idea is to follow the proofs of the previous best results in this direction and apply the main lemma at suitable points. We defer more details to the actual proofs.

	\section{Preliminaries}
	\label{sec:prelims}
	
	  We use the notation $[a,b]$ to denote an interval in $\mb{R}$ as well as an interval in $\mb{Z}$. The distinction will be clear from context.  
	  
	  \paragraph{Multilinear polynomials and Multilinearization.} Fix any field $\F$. Throughout, we work with functions $f:\{0,1\}^n\rightarrow \F$ which are represented by multilinear polynomials.  Recall that each such function has a \emph{unique} multilinear polynomial representation. Further, given a (possibly non-multlinear) polynomial $P(x_1,\ldots,x_n)$ representing $f$ (i.e. $P(a) = f(a)$ for all $a\in \{0,1\}^n$), we can obtain a multilinear representation $Q$ by simply replacing each $x_i^r$ for $r > 1$ by $x_i$ in the polynomial $P$. This preserves the underlying function as $b^r = b$ for $b\in \{0,1\}$. Any polynomial $P$ can be \emph{multilinearized} this way without increasing the degree.

	  \paragraph{Bernstein's inequality.} The following standard deviation bound can be found in, e.g., the book of Dubhashi and Panconesi~\cite[Theorem 1.2]{dubhashi_panconesi_2009}.
	\begin{lemma}[Bernstein's inequality]\label{lem:bernstein}
		Let $X_1,\ldots,X_m$ be independent and identically distributed Bernoulli random variables with mean $q$.  Let $X=\sum_{i=1}^mX_i$.  Then for any $\theta>0$,
		\[
		\prob{}{|X-mq|>\theta}\le2\exp\bigg(-\frac{\theta^2}{2mq(1-q)+2\theta/3}\bigg).
		\]
	\end{lemma}	
	
	\subsection{Symmetric Boolean functions}
	Let $n$ be a growing integer parameter which will always be the number of input variables. We use $\sB_n$ to denote the set of all symmetric Boolean functions on $n$ variables. Note that each symmetric Boolean function $f:\{0,1\}^n\rightarrow \{0,1\}$ is uniquely specified by a string $\spec f:[0,n]\rightarrow \{0,1\}$, which we call the \emph{Spectrum} of $f$, in the sense that for any $a\in \{0,1\}^n$, we have
	\[
	f(a) = \spec f(|a|).
	\]
	
	Given a $f\in \sB_n$, we define the \emph{period of $f$}, denoted $\per(f),$ to be the smallest positive integer $b$ such that $\spec f(i) = \spec f(i+b)$ for all $i\in[0,n-b]$. We say $f$ is \emph{$k$-bounded} if $\spec f$ is constant on the interval $[k,n-k]$; let $B(f)$ denote the smallest $k$ such that $f$ is $k$-bounded.
	
	\subparagraph*{Standard decomposition of a symmetric Boolean function~\cite{Lu}.} Fix any $f\in \sB_n.$ Among all symmetric Boolean functions $f'\in \sB_n$ such that $\spec f'(i) = \spec f(i)$ for all $i\in[\lceil n/3\rceil + 1,\lfloor 2n/3\rfloor],$ we choose a function $g$ such that $\per(g)$ is as small as possible. We call $g$ the \emph{periodic part} of $f$. Define $h\in \sB_n$ by $h = f\oplus g.$ We call $h$ the \emph{bounded part} of $f$. 
	
	We will refer to the pair $(g,h)$ as a \emph{standard decomposition} of the function $f$. Note that we have $f = g\oplus h.$
	
	\begin{observation}
		\label{obs:decomp}
		Let $f\in \sB_n$ and let $(g,h)$ be a standard decomposition of $f$. Then, $\per(g)\leq \lfloor n/3\rfloor$ and $B(h)\leq \lceil n/3 \rceil.$
	\end{observation}
	
	\subparagraph*{Some symmetric Boolean functions.} 
	Fix some positive $n\in \mathbb{N}$. The \emph{Majority} function $\Maj_n$ on $n$ Boolean variables accepts exactly the inputs of Hamming weight greater than $n/2.$ For $t\in [0,n]$, the \emph{Threshold} function $\Thr^t_n$ accepts exactly the inputs of Hamming weight at least~$t$; and similarly,  the \emph{Exact Threshold} function $\EThr^t_n$ accepts exactly the inputs of Hamming weight exactly $t$. Finally, for $b\in [2,n]$ and $i\in [0,b-1]$, the function $\MOD^{b,i}_n$ accepts exactly those inputs $a$ such that $|a| \equiv i\pmod{b}.$ In the special case that $i=0$, we also use $\MOD^b_n.$

	\subsection{Probabilistic polynomials}	

	\begin{definition}[Probabilistic polynomial and Probabilistic degree]
		A \emph{probabilistic polynomial} is a random  polynomial $\bm{P}$ (with some distribution having finite support) over $\F[x_1,\ldots,x_n].$ We say that the degree of $\bm{P}$, denoted $\deg(\bm{P})$, is at most $d$ if the probability distribution defining~$\bm{P}$ is supported on polynomials of degree at most $d$.
		 
		Given a Boolean function $f:\{0,1\}^n\rightarrow \{0,1\}$ and an $\varepsilon>0,$ an \emph{$\varepsilon$-error probabilistic polynomial} for $f$ is a probabilistic polynomial $\bm{P}$ such that for each $a\in \{0,1\}^n$,
		\[
		\prob{\bm{P}}{\bm{P}(a) \neq f(a)} \leq \varepsilon.
		\]
		We define the \emph{$\varepsilon$-error probabilistic degree} of $f$, denoted $\pdeg^{\F}_\varepsilon(f)$, to be the least $d$ such that $f$ has an $\varepsilon$-error probabilistic polynomial of degree at most $d$.
		
		When the field $\F$ is clear from context, we use $\pdeg_\varepsilon(f)$ instead of $\pdeg^\F_\varepsilon(f).$
	\end{definition}	
	
	\begin{fact}
		\label{fac:pdeg}
		We have the following simple facts about probabilistic degrees of Boolean functions. Let $\F$ be any field.
		\begin{enumerate}
			\item (Error reduction~\cite{HS}) For any $\delta < \varepsilon \leq 1/3$ and any Boolean function $f$, if $\bm{P}$ is an $\varepsilon$-error probabilistic polynomial for $f$, then $\bm{Q} = M(\bm{P}_1,\ldots,\bm{P}_\ell)$ is a $\delta$-error probabilistic polynomial for $f$ where $\ell = O(\log(1/\delta)/\log(1/\varepsilon)),$ $M$ is the exact multilinear polynomial for $\Maj_\ell,$ and $\bm{P}_1,\ldots,\bm{P}_\ell$ are independent copies of $\bm{P}.$ In particular, we have $\pdeg_{\delta}^\F(f) \leq \pdeg_{\varepsilon}^\F(f)\cdot O(\log(1/\delta)/\log(1/\varepsilon)).$
			\item (Composition) For any Boolean function $f$ on $k$ variables and any Boolean functions $g_1,\ldots,g_k$ on a common set of $m$ variables,  let $h$ denote the natural composed function $f(g_1,\ldots,g_k)$ on $m$ variables. Then, for any $\varepsilon, \delta > 0,$ we have $\pdeg_{\varepsilon + k\delta}^\F(h) \leq \pdeg_\varepsilon^\F(f)\cdot \max_{i\in [k]} \pdeg_\delta^\F(g_i).$ 
			\item (Sum) Assume that $f,g_1,\ldots,g_k$ are all Boolean functions on a common set of $m$ variables such that the functions $g_1,\ldots,g_k$ are mutually exclusive and $f = \sum_{i\in [k]}g_i$. Then, for any $\delta > 0,$ we have $\pdeg_{k\delta}^\F(f) \leq \max_{i\in [k]} \pdeg_\delta^\F(g_i).$
		\end{enumerate}
The first item above is not entirely obvious, as the polynomial $\bm{P}$ is not necessarily Boolean-valued at points when $\bm{P}(a)\neq f(a)$. Hence, it is not clear that composing with a polynomial that computes the Boolean Majority function achieves error-reduction. The second and third items above are trivial.
	\end{fact}
	
Building on work of Alman and Williams~\cite{AW} and Lu~\cite{Lu}, Tripathi, Venkitesh and the author~\cite{STV}  gave upper bounds on the probabilistic degree of any symmetric function. We recall below the statement in the case of fixed positive characteristic.
	
	\begin{theorem}[Known upper bounds on probabilistic degree of symmetric functions~\cite{STV}]
		\label{thm:ubd-STV}
		Let $\F$ be a field of constant characteristic $p>0$ and $n\in \mathbb{N}$ be a growing parameter. Let $f\in \sB_n$ be arbitrary and let $(g,h)$ be a standard decomposition of $f$. Then we have the following for any $\varepsilon > 0.$
		\begin{enumerate}
			\item If $\per(g) = 1,$ then $g$ is a constant and hence $\pdeg_\varepsilon(g) = 0.$ 
			
			If $\per(g)$ is a power of $p$, then $g$ can be \emph{exactly} represented\footnote{While this is not part of the formal theorem statement from~\cite{STV}, it follows readily from the proof.} as a polynomial of degree at most $\per(g)$, and hence $\pdeg^\F_\varepsilon(g) \leq \per(g),$ 
			\item $\pdeg_\varepsilon(h) = O(\sqrt{B(h)\log(1/\varepsilon)} + \log(1/\varepsilon))$ if $B(h) \geq 1$ and $0$ otherwise, and
			\item $\pdeg_\varepsilon(f)=
			\left\{
			\begin{array}{ll}
			O(\sqrt{n\log(1/\varepsilon)}) & \text{if $\per(g) > 1$ and not a power of $p$},\\
			O(\min\{\sqrt{n\log(1/\varepsilon)},\per(g)\}) & \text{if $\per(g)$ a power of $p$ and $B(h) =0$,}\\
			O(\min\{\sqrt{n\log(1/\varepsilon)},\per(g) +  & \text{otherwise.}\\
			\ \ \sqrt{B(h)\log(1/\varepsilon)} + \log(1/\varepsilon) \}) & 	
			\end{array}
			\right.$
		\end{enumerate}
	\end{theorem}

	\subsection{A string lemma}
	
	Given a function $w:I\rightarrow \{0,1\}$ where $I\subseteq \mathbb{N}$ is an interval, we think of $w$ as a string from the set $\{0,1\}^{|I|}$ in the natural way. For an interval $J\subseteq I,$ we denote by $w|_{J}$ the substring of $w$ obtained by restriction to $J$.

	The following simple lemma can be found, e.g. as a special case of ~\cite[Theorem 3.1]{word-survey}. For completeness, we give a short proof in Appendix~\ref{appsec:string-lemma}.

        \begin{restatable}{lemma}
	{lemstring}
		\label{lem:string}
		Let $w\in \{0,1\}^+$ be any non-empty string\footnote{Recall that, for any alphabet $\Sigma,$ the notation $\Sigma^+$ denotes the set of non-empty strings over this alphabet.} and $u,v\in \{0,1\}^+$ such that $w = uv = vu$. Then there exists a string $z\in \{0,1\}^+$ such that $w$ is a power of $z$ (i.e. $w= z^k$ for some $k\geq 2$).
      \end{restatable}

	\begin{corollary}
		\label{cor:string}
		Let $g\in \sB_n$ be arbitrary with $\per(g) = b > 1.$ Then for all $i,j\in [0,n-b+1]$ such that $i \not\equiv j \pmod{b}$, we have $\spec g|_{[i,i+b-1]} \neq \spec g|_{[j,j+b-1]}.$
	\end{corollary}
	
	\begin{proof}
		Suppose $\spec g|_{[i,i+b-1]} = \spec g|_{[j,j+b-1]}$ for some $i \not\equiv j \pmod{b}.$ Assume without loss of generality that $i < j < i+b.$ Let $u = \spec g|_{[i,j-1]}, v = \spec g|_{[j,i+b-1]}, w = \spec g|_{[i+b,j+b-1]}$. Then $u=w$ and the assumption $uv=vw$ implies $uv=vu$. By Lemma~\ref{lem:string}, there exists a string $z$ such that $uv = z^k$ for $k\geq 2$ and therefore $\per(g) < b$. This contradicts our assumption on $b$.
	\end{proof}
	
	\subsection{Lucas's theorem}
	
	\begin{theorem}[Lucas's theorem]
	\label{thm:lucas}
	Let $A,B$ be any non-negative integers and $p$ any prime. Then
	\[
	\binom{A}{B}= \prod_{i \geq 0} \binom{A_i}{B_i}\pmod{p}
	\]
	where $A_i$ (resp. $B_i$) is the $(i+1)$th least significant digit of $A$ (resp. $B$) in base $p$.
	\end{theorem}
	
	The following is a standard application of Lucas's theorem, essentially observed by Lu~\cite{Lu} and Heged\H{u}s~\cite{Hegedus}, showing that Heged\H{u}s's lemma is tight.
	
	\begin{corollary}
	\label{cor:lucas}
	Fix any prime $p$ and positive integer $n$. Assume $i$ is a non-negative integer and $q$ a positive integer such that $i + q\leq n.$ Let $p^\ell$ be the largest power of $p$ dividing $q$. Then, there is a symmetric multilinear polynomial $Q\in \F_p[x_1,\ldots,x_n]$ of degree $p^\ell$ such that $Q$ vanishes at all points of $\{0,1\}^n_i$ but at no point of $\{0,1\}^n_{i+q}.$
	\end{corollary}
	
	\begin{proof}
	Assume $q = p^\ell s$ where $s$ is not divisible by $p$. Let $a_\ell, b_\ell \in \{0,\ldots,p-1\}$ be the $(\ell+1)$th least significant digit of $i$ and $i+q$ respectively in base $p$. Note that $b_\ell = a_\ell + s_0\pmod{p}$ where $s_0$ is the least significant digit of $s$ in base $p$ ($s_0$ is non-zero as $s$ is not divisible by $p$).
	
	Define the polynomial 
	\[
	Q(x_1,\ldots,x_n) = \left(\sum_{S\subseteq [n]: |S| = p^{\ell}} \prod_{i\in S}x_i \right)- a_\ell,
	\]
	which we consider an element of $\F_p[x_1,\ldots,x_n].$ Note that at any input $c\in \{0,1\}^n$ of Hamming weight $w$, we have
	\[
	Q(c) = \binom{w}{p^\ell} - a_\ell
	\]
	where the right hand side is interpreted modulo $p$. Lucas's theorem then easily implies that $Q(c) = 0$ if $w =i$ and $s_0$ if $w = i+q.$
	\end{proof}

\section{The Main Lemma}

	In this section, we prove the main lemma, which is a robust version of Lemma~\ref{lem:Hegedus}. 

	\begin{lemma}[A Robust Version of Heged\H{u}s's Lemma]
	\label{lem:main}
	Assume that $\F$ is a field of characteristic~$p$. Let $n$ be a growing parameter and assume  we have positive integer parameters $k,q$ such that $100q < k < n-100q$ and $q$ is a power of $p$. Define $\alpha = \min\{k/n,1-(k/n)\}$ and $\delta = q/n.$ Assume $P\in \F[x_1,\ldots,x_n]$ is a polynomial such that for some $K \in \{k+q,k-q\}$,
	\begin{subequations}
	\label{eq:mainlem0}
	\begin{align}
	\prob{\bm{a}\sim \{0,1\}^n_{k}}{P(\bm{a}) \neq 0} &\leq \min\{e^{-100\delta^2n/\alpha},1/1000\}	\label{eq:mainlem0a}\\
	\prob{\bm{a}\sim \{0,1\}^n_{K}}{P(\bm{a}) \neq 0} &\geq e^{-\delta^2n/100\alpha}.	\label{eq:mainlem0b} 
	\end{align}
	\end{subequations}
	Then, $\deg(P) = \Omega(q),$ where the $\Omega(\cdot)$ hides an absolute constant. 
	\end{lemma}

	One can ask if the above lemma can be proved under weaker assumptions: specifically, if the upper bound in (\ref{eq:mainlem0a}) can be relaxed. It turns out that it cannot (up to changing the constant in the exponent) because for larger error parameters, there is a sampling-based construction of a polynomial with smaller degree that is zero on most of $\{0,1\}^n_k$ and non-zero on most of $\{0,1\}^n_K.$ We discuss this construction in Section~\ref{sec:tightness}.
	
	We  first prove a special case of the lemma which corresponds to the case when $K = k+q = \lfloor n/2\rfloor$ and $q$ sufficiently larger than $\sqrt{n}.$ This case suffices for most of our applications. The general case is a straightforward reduction to this special case.
	
	\subsection{A special case}
	\label{sec:specialcase}

	\begin{lemma}[A special case of Lemma~\ref{lem:main}]
	\label{lem:Dij}
	Let $n$ be a growing parameter and assume $\varepsilon \in [2^{-n/100},e^{-200}].$ Assume $t$ is an integer such that $t$ is a power of $p$ and furthermore, $t = \sqrt{n\ell}$ for some $\ell\in \mathbb{R}$ such that $100\leq \ell\leq \frac{1}{2}\cdot \ln(1/\varepsilon).$ Let $P\in \F[x_1,\ldots,x_n]$ be any polynomial such that 
	\begin{subequations}
	\begin{align}
	\prob{\bm{a}\sim \{0,1\}^n_{\lfloor n/2\rfloor - t}}{P(\bm{a}) \neq 0} &\leq \varepsilon	\label{eq:n/2-t}\\
	\prob{\bm{a}\sim \{0,1\}^n_{\lfloor n/2\rfloor}}{P(\bm{a}) \neq 0} &\geq e^{-\ell/2} \label{eq:n/2}.
	\end{align}
	\end{subequations}
	Then, $\deg(P) \geq t/25.$
	\end{lemma}
	
	\begin{remark}
	\label{rem:negation}
	By negating inputs (i.e. replacing $x_i$ with $1-x_i$ for each $i$), the above lemma also implies the analogous statements where $\lfloor n/2\rfloor - t$ and $\lfloor n/2\rfloor$ are replaced by $\lceil n/2\rceil+t$ and $\lceil n/2 \rceil$ respectively.
	\end{remark}
	
		Before we prove this lemma, we need to collect some technical facts and lemmas.
		
		The following is standard. See, e.g.,~\cite[Lemma 3.3]{KS} for a proof.
		\begin{fact}
		\label{fac:int-set}
		Let $R\in \F[x_1,\ldots,x_n]$ be a non-zero multilinear polynomial of degree at most $d\leq n$. Then $R$ cannot vanish at all points in any Hamming ball of radius $d$ in $\{0,1\}^n$.
		\end{fact}
		
		\begin{lemma}
		\label{lem:technical-binom}
		Let $n, r,s$ be any non-negative integers with $r\leq s \leq n/4.$ Then we have
		\[
		e^{-8s(r-s)/n}\leq \frac{\binom{n}{\lfloor n/2\rfloor - s}}{\binom{n}{\lfloor n/2\rfloor - r}}\leq e^{-2r(r-s)/n}.
		\]
		\end{lemma}
		
		\begin{proof}		
		Note that 
		\[
		\frac{\binom{n}{\lfloor n/2\rfloor - s}}{\binom{n}{\lfloor n/2\rfloor - r}} = \frac{(\lfloor n/2\rfloor - s+1)\cdots (\lfloor n/2\rfloor - r)}{(\lceil n/2\rceil +s)\cdots (\lceil n/2\rceil + r+1)} \leq \left(\frac{\lfloor n/2\rfloor - r}{\lceil n/2\rceil + r}\right)^{r-s}\leq \left(1-\frac{2r}{n}\right)^{r-s}\leq e^{-2r(r-s)/n},
		\]
		which implies the right inequality in the statement of the claim. We have used the inequality $1-x\leq e^{x}$ to deduce the final inequality above.
		
		For the left inequality, we similarly have
		\begin{align*}
		\frac{\binom{n}{\lfloor n/2\rfloor - s}}{\binom{n}{\lfloor n/2\rfloor - r}} \geq  \left(\frac{\lceil n/2\rceil - s}{\lceil n/2\rceil + s}\right)^{r-s} \geq \left(\left(1-\frac{2s}{n}\right)^{2}\right)^{r-s} \geq e^{-8s(r-s)/n}.
		\end{align*}
		where the final inequality follows from the fact that $(1-x)\geq e^{-2x}$ for $x\in [0,1/2]$.
		\end{proof}

		Given a set $E\subseteq \{0,1\}^n,$ and a parameter $D\leq n$, we define $\mc{I}_D(E)$ to be the set of all multilinear polynomials $Q$ of degree at most $D$ that vanish at all points of $E$. Further, we define the \emph{degree-$D$ closure of $E$}, denoted $\mathrm{cl}_D(E)$ as follows.
		\[
		\mathrm{cl}_D(E) := \{a \in \{0,1\}^n\ |\ Q(a) = 0\ \  \forall Q\in \mc{I}_D(E)\}.
		\]
		Note that $\mathrm{cl}_D(E)\supseteq E$ but could be much bigger than $E$. The following result of Nie and Wang~\cite{NW} gives a bound on $|\mathrm{cl}_D(E)|$ in terms of $|E|.$ (This particular form is noted and essentially proved in~\cite{NW}, and is explicitly stated and proved  in~\cite[Theorem A.1]{KS} for \emph{all} fields.)
		
		\begin{theorem}
		\label{thm:NW}
		For any $E\subseteq \{0,1\}^n$ and any $D\leq n$, we have
		\[
		\frac{\mathrm{cl}_D(E)}{2^n}\leq \frac{|E|}{N_D}
		\]
		where $N_D = \sum_{j=0}^D \binom{n}{j},$ the number of multilinear monomials of degree at most $D$.
		\end{theorem}
		
		\begin{remark}
		\label{rem:NW}
		It should be noted that the above lemma generalizes the standard linear-algebraic fact that for any $E$ such that $|E| < N_D$, there is a non-zero multilinear polynomial of degree $D$ that vanishes on $E$. Or equivalently,
		\[
		|E| < N_D \Longrightarrow \mathrm{cl}_D(E) < 2^n.
		\]
		The inequality stated in the lemma is tight for certain sets $E$ of size $N_D$ (a good example of such a set is any Hamming ball of radius $D$). However, when $|E|$ is much smaller than $N_D$, the parameters can be tightened. A tight form of this lemma, that gives the best possible parameters depending on $|E|$, was proved in earlier work of Keevash and Sudakov~\cite{KeevashSudakov} (see also the works of Clements and Lindstr\"{o}m~\cite{CL}, Wei~\cite{Wei}, Heijnen and Pellikaan~\cite{HP}, and Beelen and Dutta~\cite{BD} that prove similar results). However, we don't need this general form of the lemma here.
		\end{remark}
	
	We now begin the proof of the Lemma~\ref{lem:Dij}.
	
	\begin{proof}[Proof of Lemma~\ref{lem:Dij}]
	Assume that $P$ is as given. Let $m = \lfloor n/2\rfloor$.
	
	Let $E_0,E_1$ be defined as follows. (Here, the notation ``$E$'' stands for ``error sets''.)
	\begin{align*}
	E_0 &= \{a\in \{0,1\}^n_{m-t}\ |\ P(a)\neq 0\}\\
	E_1 &= \{a\in \{0,1\}^n_m\ |\ P(a) = 0\}
	\end{align*}
	
	We show that there are polynomials $Q_1,Q_2\in\F[x_1,\ldots,x_n]$ such that the following conditions hold.
	\begin{description}
	\item{(Q1.1)} $Q_1(a) \neq 0$ if and only if $|a|\equiv m \pmod{t}.$
	\item{(Q2.1)} $Q_2(a) = 0$ for all $a\in E_0$.
	\item{(Q2.2)} $Q_2(a) = 0$ for all $a$ such that $|a| < m-t$ and $|a|\equiv m \pmod{t}.$
	\item{(Q2.3)} $Q_2(a) \neq 0$ for some $a\in \{0,1\}^n_m \setminus E_1.$
	\end{description}

	Given polynomials $Q_1,Q_2$ as above, we construct the polynomial $R$ to be the multilinear polynomial obtained by computing the formal product $P\cdot Q_1\cdot Q_2$ and replacing $x_i^r$ by $x_i$ for each $r > 1$. Note that $R(a) = P(a)Q_1(a)Q_2(a)$ for any $a\in \{0,1\}^n$.
	
	We observe that $R(a) = 0$ for all $|a| < m.$ This is based on a case analysis of whether $|a| \equiv m \pmod{t}$ or not. In the latter case, we see that $Q_1(a) = 0$ and hence $R(a) = 0.$ In the former case, we have either $a\in \{0,1\}^n_{m-t}\setminus E_0$, in which case $P(a) = 0$, or not, in which case $Q_2(a) = 0.$ Hence, $R(a) = 0$ for all $|a| < m.$
	
	On the other hand, we note that $R$ is a non-zero polynomial. This is because by (Q2.3), we know that there is some $a'\in \{0,1\}^n_m\setminus E_1$ where $Q_2(a') \neq 0.$ Further, $Q_1(a')\neq 0$ and $P(a')\neq 0$ by (Q1.1) and the definition of $E_1$ respectively. Hence, $R(a')\neq 0,$ implying that $R$ is a non-zero multilinear polynomial. 
	
	By Fact~\ref{fac:int-set}, we thus know that $R$ has degree at least $m$. In particular, we obtain
	\begin{equation}
	\deg(P) \geq \deg(R) - \deg(Q_1) - \deg(Q_2) \geq m- \deg(Q_1) - \deg(Q_2).\notag
	\end{equation}
	
	Hence, to finish the proof of the lemma, it suffices to prove the following claims.
	
	\begin{claim}
	\label{clm:Q1}
	There is a $Q_1$ of degree at most $t$ satisfying property (Q1.1).
	\end{claim}
	
	\begin{claim}
	\label{clm:Q2}
	There is a $Q_2$ of degree at most $m-t-t_1$ satisfying properties (Q2.1)-(Q2.3), where $t_1 = \lceil t/25\rceil.$
	\end{claim}
	
	We now prove the above claims.
	
	\begin{subproof}[Proof of Claim~\ref{clm:Q1}]
	This follows immediately from the upper bound for periodic functions in Theorem~\ref{thm:ubd-STV}. Consider the $t$-periodic function that takes the value $1$ at point $a\in \{0,1\}^n$ if and only if $|a|\equiv m\pmod{t}.$ Since this function is $t$-periodic, it can be represented exactly as a polynomial of degree at most $t$. This yields the claim.
	\end{subproof}

	\begin{subproof}[Proof of Claim~\ref{clm:Q2}]
	Let $D$ denote $m-t-t_1.$ Let $E = E_0 \cup \bigcup_{j < m-t: j\equiv m \pmod{t}}\{0,1\}^n_j.$ We want to show the existence of a polynomial $Q_2$ of degree at most $D$ such that $Q_2$ vanishes at all points of $E$ but $Q_2$ does not vanish at some point in ${E}_1' := \{0,1\}^n_m\setminus E_1$. Note that this is equivalent to saying that $\mathrm{cl}_D(E)\nsupseteq {E}_1'$. To show this, it suffices to show that 
	\begin{equation}
	\label{eq:clmQ2.1}
	|\mathrm{cl}_D(E)| <  e^{-\ell/2}\cdot \binom{n}{m}
	\end{equation}
	since by hypothesis we have $|{E}_1'|\geq e^{-\ell/2}\cdot \binom{n}{m}.$
	
	To do this, we use Theorem~\ref{thm:NW}. Note that we have
	\begin{align*}
	|E| &\leq |E_0| + \sum_{j < m-t: j\equiv m \pmod{t}} \binom{n}{j}\\
	&\leq \varepsilon\cdot \binom{n}{m-t} + \sum_{k\geq 1} \binom{n}{m-t-k\cdot t}\\
	&\leq \varepsilon\cdot \binom{n}{m-t} + \binom{n}{m-t}\cdot \left(e^{-2\ell} + e^{-4\ell} + \ldots\right)\\
	&\leq \binom{n}{m-t}\cdot (\varepsilon + 2\cdot e^{-2\ell}) \leq \binom{n}{m-t}\cdot (3e^{-2\ell})\numberthis \label{eq:clmQ2.2}
	\end{align*}
	where the third inequality is a consequence of Lemma~\ref{lem:technical-binom} (with $r = t$ and $s = (k+1)t$ for various~$k$) and the final inequality uses $\varepsilon\leq e^{-2\ell}.$
	
	On the other hand, the parameter $N_D$ from the statement of Theorem~\ref{thm:NW} can be lower bounded as follows.
	\begin{align*}
	N_D &= \sum_{j=0}^D\binom{n}{D-j} \geq t_1 \binom{n}{m-t-2t_1}\\
	&\geq t_1 e^{-\ell}\cdot \binom{n}{m-t} > e^{-\ell}\cdot \frac{\sqrt{n}}{3}\cdot \binom{n}{m-t}
	\end{align*}
	where the second inequality follows from Lemma~\ref{lem:technical-binom} (with $r = t$ and $s = t+2t_1$) and the final inequality uses the fact that $t_1 > t/30 = \sqrt{n\ell}/30 \geq \sqrt{n}/3.$
	
	Putting the above together with (\ref{eq:clmQ2.2}) immediately yields
	\[
	\frac{|E|}{N_D} < 9e^{-\ell}\cdot \frac{\binom{n}{m-t}}{\sqrt{n}\cdot\binom{n}{m-t}} =9e^{-\ell}\cdot n^{-1/2}.
	\]
	Using Theorem~\ref{thm:NW}, we thus obtain
	\[
	\mathrm{cl}_D(E) < 9e^{-\ell}\cdot \frac{2^n}{\sqrt{n}}\leq e^{-\ell/2}\cdot  \frac{2^n}{2\sqrt{n}} \leq e^{-\ell/2}\cdot\binom{n}{m}
	\]
	where the last inequality follows from Stirling's approximation. Having shown (\ref{eq:clmQ2.1}), the claim now follows.
	\end{subproof}
	\end{proof}
	
	\subsection{The General Case}
	\label{sec:generalcase}
	
	We start with some preliminaries.

	We first show a simple `error-reduction' procedure for polynomials. For any polynomial $P\in \F[x_1,\ldots,x_n]$ and any $m\in [0,n]$, let  $\mathrm{NZ}_m(P)$ denote the set of points of $\{0,1\}^n_m$ where $P$ does not vanish. Let $\psi_m(P)$ denote $|\mathrm{NZ}_m(P)|/\binom{n}{m}.$
	
	\begin{lemma}
	\label{lem:repetition}
	 For any $Q\in \F[x_1,\ldots,x_n]$ and any $r\geq 1$, there is a probabilistic polynomial $\bm{Q^{(r)}}$ of degree at most $r\cdot \deg(Q)$ such that for all $m\in [0,n]$, $\avg{\bm{Q^{(r)}}}{\psi_m(\bm{Q^{(r)}})} = \psi_m(Q)^r.$
	\end{lemma}
	
	\begin{proof}
	For a permutation $\pi\in S_n,$ and $a\in \{0,1\}^n,$ define $a^\pi = (a_{\pi(1)},\ldots, a_{\pi(n)}).$ Also, define $Q^\pi(x_1,\ldots,x_n) = Q(x^\pi) = Q(x_{\pi(1)},\ldots,x_{\pi(n)}).$
	
	For a uniformly random $\bm{\pi}\in S_n,$ and any $a\in \{0,1\}^n_m$, the probabilistic polynomial $Q^{\bm{\pi}}$ satisfies 
	\[
	\prob{\bm{\pi}}{Q^{\bm{\pi}}(a) \neq 0} = \prob{\bm{\pi}}{Q(a^{\bm{\pi}}) \neq 0} =  \prob{\bm{\pi}}{a^{\bm{\pi}}\in \mathrm{NZ}_m(Q)} = \psi_m(Q)
	\]
	as $a^{\bm{\pi}}$ is uniformly distributed over $\{0,1\}^n_m.$
	
	Choose $\bm{\pi_1},\ldots,\bm{\pi_r}$ i.u.a.r. from $S_n,$ and define $\bm{Q^{(r)}} = \prod_{i=1}^r Q^{\bm{\pi_i}}.$ For any $a\in \{0,1\}^n_m$
	\[
	\prob{\bm{Q^{(r)}}}{\bm{Q^{(r)}}(a) \neq 0} = (\psi_m(Q))^r.
	\]
	In particular, the above holds for a uniformly random $\bm{a}$ chosen from $\{0,1\}^n_m.$ Hence, we have 
	\[
	\avg{\bm{Q^{(r)}}}{\psi_m(\bm{Q^{(r)}})} = \prob{\bm{Q^{(r)}}, \bm{a}\sim \{0,1\}^n_m}{\bm{Q^{(r)}}(\bm{a}) \neq 0} = \psi_m(Q)^r.
	\]
	\end{proof}
	
	We are now ready to prove the main lemma in its full generality.
	
	\begin{proof}[Proof of Lemma~\ref{lem:main}]
	W.l.o.g. we assume that $k\leq n/2$. (To prove the lemma for $k > n/2$, consider the polynomial $Q(x) = P(1-x_1,\ldots,1-x_n)$ instead.)
	
	We first reduce to the case where $K = n/2$. 
	
	More precisely, note that there exist non-negative integers $r\leq 2q$ and $s$ so that $2(K-r) = n-r-s.$  This can be seen by a simple case analysis. If $K = k-q$, we can choose $r=0$, $s = n-2k + 2q$; if $K = k+q$ and $n-2k\geq 2q$, we can choose $r=0$ and $s = n-2k -2q$; and if $K = k+q$ and $n-2k < 2q$, we can choose $r = 2q-(n-2k)$ and $s = 0$.
	
	Having chosen $r,s$ as above, we set $K' = K-r$, $k' = k-r$ and $n' = n-r-s$. Let $\bm{S}$ be a uniformly random subset of $[n]$ of size $r+s$ and $\bm{y}$ a uniformly random point in $\{0,1\}^{r+s}_r$. We set $P_{\bm{S},\bm{y}}(x_i: i\not\in S)$ to be the probabilistic polynomial obtained by setting all the variables indexed by $\bm{S}$ according to $\bm{y}$. Note that we have
	\[
	\avg{\bm{S},\bm{y}}{\psi_{k'}(P_{\bm{S},\bm{y}})} = \psi_k(P) =: \varepsilon_0 \phantom{xxx}\text{ and }\phantom{xxx} \avg{\bm{S},\bm{y}}{\psi_{K'}(P_{\bm{S},\bm{y}})} = \psi_K(P) =: \varepsilon_1.
	\]
	By Markov's inequality, we have
	\[
	\prob{\bm{S},\bm{y}}{\psi_{k'}(P_{\bm{S},\bm{y}}) > \frac{2\varepsilon_0}{\varepsilon_1}} < \frac{\varepsilon_1}{2} \phantom{xxx}\text{ and }\phantom{xxx} \prob{\bm{S},\bm{y}}{\psi_{K'}(P_{\bm{S},\bm{y}}) > \frac{\varepsilon_1}{2}} \geq \frac{\varepsilon_1}{2}.
	\]
	Hence, with positive probability over the choice of $\bm{S}$ and $\bm{y},$ we have both $\psi_{k'}(P_{\bm{S},\bm{y}}) \leq 2\varepsilon_0/{\varepsilon_1}$ and $\psi_{K'}(P_{\bm{S},\bm{y}}) > {\varepsilon_1}/{2}.$ We fix such a choice $S,y$ for $\bm{S},\bm{y}$ and let $P'$ denote $P_{S,y}.$ Clearly, $\deg(P) \geq \deg(P')$ and hence it suffices to lower bound $\deg(P').$ 
	
	We will now use Lemma~\ref{lem:Dij} to obtain the desired lower bound on $\deg(P').$ First of all, note that $\ell' := q^2/n'$ satisfies
	\[
	\ell' = \frac{q^2}{n'}\leq \frac{k^2}{10000 n'} \leq \frac{n'}{10000},
	\]
	by the bounds on $q$ in the statement of the lemma and the fact that $k\leq 2K = n'$. 
	
	We consider now two cases. 
	
	\paragraph{Case 1:} Assume first that $\ell'\geq 100.$ Using the bounds on $\varepsilon_0$ and $\varepsilon_1$ that follow from the lemma statement and the bounds above,  $P'$ is a polynomial in $n'$ variables satisfying
	\begin{align*}
	\psi_{k'}(P') &\leq \frac{2\varepsilon_0}{\varepsilon_1}\leq 2\varepsilon_0^{0.99}\leq  2\exp(-99 \delta^2 n/\alpha) = 2\exp(-99\delta^2n^2/(\alpha n)) \leq 2\exp(-99q^2/n'), \text{  and } \\
	\psi_{K'}(P') &\geq  \frac{\varepsilon_1}{2}\geq  \frac{1}{2}\exp(-(1/100)\cdot \delta^2 n/\alpha)=  \frac{1}{2}\exp(-(1/100)\cdot \delta^2 n^2/(\alpha n))\\
 &\geq \frac{1}{2} \exp(-(1/25)\cdot q^2/n') .
	\end{align*}
	where  we have used the inequalities $n' \geq 2(k-q) \geq \alpha n$ and $n' = 2K' \leq 2(k+q) \leq 4\alpha n.$

	Define $\varepsilon = \exp(-2\ell').$ Note that we have $\varepsilon \geq \exp(-n'/5000)$ by the bound on $\ell'$ above. Further,
	\begin{align*}
	\psi_{(n'/2) - q}(P') &= \psi_{k'}(P') \leq 2\exp(-99q^2/n') = 2\exp(-99\ell') \leq \exp(-2\ell') = \varepsilon, \text{ and }\\
	\psi_{n'/2}(P') &= \psi_{K'}(P') \geq \frac{1}{2} \exp(-(1/25)\cdot q^2/n') \geq \exp(-\ell'/2).
	\end{align*}
	Applying Lemma~\ref{lem:Dij} to $P'$ (see also Remark~\ref{rem:negation}), we immediately obtain $\deg(P') \geq q/25$ and hence we are done in this case.
	
	\paragraph{Case 2:} Now consider the case when $\ell'< 100$. In this case, the hypothesis of the lemma assures us that $\varepsilon_0 \leq 1/1000$ and $\varepsilon_1 \geq \exp(-q^2/100\alpha n) \geq \exp(-\ell'/25) \geq e^{-4}$ where the second inequality uses $n'\leq 4\alpha n$ as argued above. Then, we have
	\begin{subequations}
	\label{eq:mainlem2}
	\begin{align}
	\psi_{k'}(P') &\leq \frac{2\varepsilon_0}{\varepsilon_1}\leq 2\varepsilon_0^{0.99}\leq \frac{1}{400},\label{eq:mainlem2a}\\
	\psi_{K'}(P') &\geq  \frac{\varepsilon_1}{2} \geq \frac{\varepsilon_0^{0.01}}{2}\geq\frac{1}{2^{100/99}}\cdot \psi_{k'}(P')^{1/99}\geq  \psi_{k'}(P')^{1/7},\label{eq:mainlem2b}\\
	\psi_{K'}(P') &\geq  \frac{\varepsilon_1}{2} \geq e^{-5}.\label{eq:mainlem2c}
	\end{align}
	\end{subequations}
	where (\ref{eq:mainlem2b}) uses $\varepsilon_0^{0.01}\geq (\psi_{k'}(P')/2)^{1/99}$ and $\psi_{k'}(P')\leq 1/400$, both of which follow from (\ref{eq:mainlem2a}).
	
	Let $r$ be a large constant that will be fixed below. By Lemma~\ref{lem:repetition}, we know that there is a probabilistic polynomial $\bm{P'^{(r)}}$ of degree at most $r\cdot \deg(P')$ such that for each $m\in \{k',K'\}$, we have $\avg{\bm{P'^{(r)}}}{\psi_m(\bm{P'^{(r)}})} = \psi_m(P')^r.$ 
	
	The proof will proceed by another restriction to $n''$ variables, where $n''$ is defined to be the largest even integer such that $100 n'' \leq q^2.$  We assume that $n''$ is greater than a large enough absolute constant, since otherwise $q$ is upper bounded by a fixed constant, in which case the degree bound to be proved is trivial. Note that $\ell'' := q^2/n'' \geq 100$ by definition. We also have $n'' = (q^2/100) - 2,$ which implies that $\ell'' \leq 100 + O(1)/q^2 \leq 101$, as long as $q$ is greater than a large enough absolute constant.
	
	Relabel the variables so that $P'$ is a polynomial in $x_1,\ldots,x_{n'}.$ Let $\bm{T}$ be a uniformly random subset of $[n']$  of size $n'-n''$ and let $\bm{z}$ be a uniformly random point in $\{0,1\}^{n'-n''}_{(n'-n'')/2}.$ Define the probabilistic polynomial $\bm{P'^{(r)}}_{\bm{T},\bm{z}}$ obtained by setting the variables indexed by $\bm{T}$ according to $\bm{z}$ in the probabilistic polynomial $\bm{P'^{(r)}}.$ Let $K'' := n''/2$ and $k'' := k'-(n'-n'')/2$. As above, we have 
	\[
	\avg{\bm{P'^{(r)}},\bm{T},\bm{z}}{\psi_{k''}(\bm{P'^{(r)}}_{\bm{T},\bm{z}})} = \psi_{k'}(P')^r =: \varepsilon_0' \phantom{xxx}\text{ and }\phantom{xxx} \avg{\bm{P'^{(r)}},\bm{T},\bm{z}}{\psi_{K''}(\bm{P'^{(r)}}_{\bm{T},\bm{z}})} = \psi_{K'}(P')^r=: \varepsilon_1'.
	\]
	
	Let $r$ be the smallest positive integer so that $\varepsilon_0'= \psi_{k'}(P')^r \leq e^{-300}.$ Note that $r$ is upper bounded by an absolute constant, as $\psi_{k'}(P')\leq 1/400$ by (\ref{eq:mainlem2a}). Further, we have $\psi_{k'}(P')^{r-1} > e^{-300}$ and hence
	\[
	\varepsilon_1' = \psi_{K'}(P')^r =\psi_{K'}(P')^{r-1}\cdot \psi_{K'}(P') \geq \left((\psi_{k'}(P'))^{r-1}\right)^{1/7}\cdot e^{-5} > e^{-48}
	\]
	where the first inequality uses (\ref{eq:mainlem2}).
	
	By Markov's inequality as above, there is a fixed choice of $\bm{P'^{(r)}}, \bm{T},$ and $\bm{z}$ such that the corresponding polynomial $P''$ is a polynomial on $n''$ variables satisfying
	\[
	\psi_{k''}(P'') \leq \frac{2\varepsilon_0'}{\varepsilon_1'} < e^{-210} < e^{-2\ell''} \phantom{xxx}\text{ and }\phantom{xxx} \psi_{K''}(P'') \geq \frac{\varepsilon_1'}{2} > e^{-50} \geq e^{-\ell''/2}.
	\]
	Applying Lemma~\ref{lem:Dij} to $P''$ with error parameter $\varepsilon = \frac{2\varepsilon_0'}{\varepsilon_1'} $ yields $\deg(P'') \geq q/25.$ As $\deg(P'')\leq r\cdot \deg(P'),$ we also get $\deg(P') = \Omega(q)$, finishing the proof in this case as well. (Note that the $\Omega(\cdot)$ hides an absolute constant.)
	\end{proof}
	
	\subsection{Tightness of the Main Lemma (Lemma~\ref{lem:main})}
	\label{sec:tightness}
	
	In this section, we discuss the near-optimality of Lemma~\ref{lem:main} w.r.t. to the various parameters. Fix $n,k,q,\alpha,\delta$ and $\F$ as in the statement of Lemma~\ref{lem:main}. Assume that $K = k+q$ (the case when $K = k-q$ is similar) and that $k\leq n/2$. Let $\varepsilon\in (0,1)$ be arbitrary. 
	
	First of all, we note that the degree lower bound obtained cannot be larger than $q$, because by Corollary~\ref{cor:lucas}, it follows that there is a degree-$q$ polynomial that vanishes at all points of weight $k$ but no points of weight $K$. 
	
	So, the statement of Lemma~\ref{lem:main} proves a lower bound on the degree that nearly (up to constant factors) matches this trivial upper bound, under the weaker assumption that the polynomial is forced to be zero only on most (say a $1-\varepsilon$ fraction) of $\{0,1\}^n_k$ and non-zero on most (say a $1-\varepsilon$ fraction) of $\{0,1\}^n_K$. (Lemma~\ref{lem:main} is a stronger statement, but we will show that even this weaker statement is tight.)
	
	In this section, we show that the value of $\varepsilon$ cannot be increased beyond $\varepsilon = \exp(-O(\delta^2n/\alpha))$, if we want to prove a lower bound of $\Omega(q)$ on the degree. More precisely, we show the following.
	
	\begin{theorem}
	\label{thm:tight}
	Assume that $\varepsilon = \exp(-o(\delta^2n/\alpha)).$ Then, there is a polynomial $P$ of degree $o(q)$ such that 
	\begin{align*}
	\prob{\bm{a}\sim \{0,1\}^n_{k}}{P(\bm{a}) \neq 0} &\leq \varepsilon\\
	\prob{\bm{a}\sim \{0,1\}^n_{K}}{P(\bm{a}) \neq 0} &\geq 1-\varepsilon.	
	\end{align*}
	\end{theorem}
	
	\begin{proof}
		To prove this theorem, we analyze a different polynomial construction to achieve this based on sampling.  We will need the following interpolation lemma that can be found in a paper of Alman and Williams~\cite{AW}.\footnote{This lemma has a trivial proof via univariate polynomial interpolation if we only want the polynomial $Q$ to have rational coefficients. However, here it important that $Q$ has integer coefficients.}
	\begin{lemma}
	\label{lem:AW}
	Let $n$ be arbitrary and $I\subseteq [0,n]$ be any interval of integers. Given any $f:I\rightarrow \{0,1\}$, there is a multilinear polynomial $Q\in \mathbb{Z}[x_1,\ldots,x_n]$ of degree at most $|I|-1$ such that $Q(a) = f(|a|)$ for each $a\in \bigcup_{i\in I}\{0,1\}^n_i$.
	\end{lemma}
	
	 Fix any positive integer $m$. By Lemma~\ref{lem:AW}, it follows that there is a multilinear polynomial $Q\in \mathbb{Z}[y_1,\ldots,y_m]$ of degree $O(\delta m)$ such that $Q(b) = 0$ for each $b\in \{0,1\}^m$ such that $|b|\in ((\alpha-\delta/2)m, (\alpha+\delta/2)m)$ and $Q(b) = 1$ for each $b\in \{0,1\}^m$ such that $|b|\in ((\alpha+\delta/2)m, (\alpha+3\delta/2)m)$. Reducing the coefficients modulo $p$, we obtain a polynomial $\tilde{Q}\in \F[y_1,\ldots,y_m]$ with the same property. Fix this $\tilde{Q}.$
	
	Consider the probabilistic polynomial $\bm{P}(x_1,\ldots,x_n)$ defined as follows. Choose $\bm{i_1},\ldots,\bm{i_m}$ i.u.a.r. from $[n]$ where $m = C \cdot (\alpha/\delta^2)\log(1/\varepsilon)$ for a large enough constant $C$ we will fix below. We define $\bm{P}(x_1,\ldots,x_n)$ to be the polynomial $\tilde{Q}(x_{\bm{i_1}},\ldots,x_{\bm{i_m}})$. Note that 
	
	\[
	\deg(\bm{P}) \leq \deg(Q) = O(\delta m ) = O((\alpha/\delta)\log(1/\varepsilon)) = o(\delta n) = o(q)
	\]
	where the second-last equality uses our assumption that $\varepsilon = \exp(-o(\delta^2 n/\alpha)).$
	
	Let $a\in \{0,1\}^n_{k}$ be arbitrary. We analyze the random variable $\bm{P}(a)$. Note that as long as the Hamming weight of $\bm{b} = (a_{\bm{i_1}},\ldots,a_{\bm{i_m}})$ is in the interval $((\alpha-\delta/2)m, (\alpha+\delta/2)m)$, we have $\bm{P}(a) = 0.$ As each co-ordinate of $\bm{b}$ is $1$ with probability $k/n = \alpha\in [0,1/2]$, Bernstein's inequality (Lemma~\ref{lem:bernstein}) yields
	\begin{align*}
	\prob{\bm{P}}{\bm{P}(a) \neq 0} &\leq \prob{\bm{i_1},\ldots,\bm{i_m}}{||\bm{b}|-\alpha m| > \delta m/3}\leq \exp(-\Omega(\delta^2 m/\alpha)) < \varepsilon/2
	\end{align*}
	as long as $C$ is a large enough constant. In a similar way, we also see that for any $a\in \{0,1\}^n_{K}$, we have $\prob{\bm{P}}{\bm{P}(a) \neq 1} < \varepsilon/2$ and hence, in particular, $\prob{\bm{P}}{\bm{P}(a) \neq 0} >1- (\varepsilon/2)$, as long as $C$ is a large enough constant.
	
	In particular, by Markov's inequality and the union bound, we see that there is a $P$ of degree at most $\deg(\bm{P})$ such that
	\[
	\psi_k(P) \leq \varepsilon \phantom{xxxx}\text{ and }\phantom{xxxx} \psi_K(P) \geq 1-\varepsilon.
	\]
	Thus, we  have a polynomial $P$ as claimed in Theorem~\ref{thm:tight}.
	\end{proof}
	
	\subsection{An extension to the case when \texorpdfstring{$q$}{q} is not a power of \texorpdfstring{$p$}{p}}
	
	An anonymous reviewer suggested the following extension of the main lemma (Lemma~\ref{lem:main}). We prove this by a simple reduction to the main lemma. (This leads to a  worsening in the constants involved.)
	
	\begin{lemma}[An extension to the case when $q$ is not a power of $p$]
	\label{lem:main-extension}
	Assume that $\F$ is a field of characteristic $p$. Let $n$ be a growing parameter and assume  we have positive integer parameters $k,q$ such that $200q < k < n-200q$.  Let $q'$ be the largest power of $p$ that divides $q$ and assume $q = q's$. Define $\alpha = \min\{k/n,1-(k/n)\}$ and $\delta = q/n$. Assume that $Q\in \F[x_1,\ldots,x_n]$ is a polynomial such that for some $K \in \{k+q,k-q\}$,
	\begin{subequations}
	\label{eq:mainlem-extension}
	\begin{align}
	\prob{\bm{a}\sim \{0,1\}^n_{k}}{Q(\bm{a}) \neq 0} &\leq \min\{e^{-1000\delta^2n/s\alpha},1/2000\}	\label{eq:mainlem-extension-a}\\
	\prob{\bm{a}\sim \{0,1\}^n_{K}}{Q(\bm{a}) \neq 0} &\geq e^{-\delta^2n/1000s\alpha}.	\label{eq:mainlem-extension-b} 
	\end{align}
	\end{subequations}
	Then, $\deg(Q) = \Omega(q'),$ where the $\Omega(\cdot)$ hides an absolute constant. 
	\end{lemma}
	
	\begin{remark}
	\label{rem:mainlem-extend-errorfree}
	The `non-robust' version of this lemma (when $Q$ vanishes everywhere on $\{0,1\}^n_k$ but not on some point in $\{0,1\}^n_K$) yields a degree lower bound of $q'$, and can be proved using similar techniques to those used in proving Heged\H{u}s's lemma. A proof can be found in~\cite{SV-arxiv}.
	\end{remark}
	
	\begin{remark}
	\label{rem:mainlem-extend-tight}
	As in the case of the main lemma, the degree lower bound obtained above is tight, using the same reasoning as in Section~\ref{sec:tightness}.
	\end{remark}
	
	\begin{proof}
	W.l.o.g. assume $K = k+q$. 
	
	Let $k' = \lfloor k/s\rfloor$ and $n' = \lfloor n/s\rfloor - 1$. Our aim will be to show using the polynomial $Q$ that there is a polynomial $P$ on $n'$ variables that distinguishes between Hamming weights $k'$ and $K' := k'+q'$. We will then appeal to Lemma~\ref{lem:main} to get the degree lower bound. 
	
	It is easy to check that $100q' < k' < n'-100q'$ as
	\begin{align*}
	100 q's = 100 q &< k - q < (k'+1)s - q \leq k's \\
	k's \leq k < n - 102 q &< (n'+2)s - 102q \leq n's - 100q = (n'-100q')s
	\end{align*}
	where we used the hypotheses that $200q < k < n-200q.$
	
	We construct the polynomial $P$ as follows. 	Assume that $k = k's + r_1$ and $n = n's+ s + r_2$ for $r_1,r_2\in \{0,\ldots,s-1\}.$ On an input $x\in \{0,1\}^{n'}$, we consider the \emph{random} input $\bm{y}\in \{0,1\}^n$ defined as follows.
	\begin{itemize}
	\item Each co-ordinate of $x$ is repeated $s$ times to get an $X\in \{0,1\}^{sn'}$.
	\item We concatenate $X$ with the string $1^{r_1} 0^{s+r_2-r_1}$ to get a string $Y\in \{0,1\}^{n}.$
	\item A uniformly random permutation $\bm{\pi}$ is applied to the $n$ coordinates of $Y$ to get $\bm{y}.$
	\end{itemize}
	Finally, we define the probabilistic polynomial $\bm{P}(x) := Q(\bm{y}).$ For a fixed permutation $\bm{\pi}$, each coordinate of $\bm{y}$ is a polynomial of degree at most $1$ in the variables $x_1,\ldots, x_{n'},$ and hence, $\deg(\bm{P}) \leq \deg(Q).$ We will show that there is some polynomial $P$ in the support of $\bm{P}$ that has the desired properties.
	
	Let $\varepsilon_0$ and $\varepsilon_1$ denote the right hand sides of inequalities (\ref{eq:mainlem-extension-a}) and (\ref{eq:mainlem-extension-b}) respectively. Observe that when $x\in\{0,1\}^{n'}_w$, then the random $\bm{y}\in \{0,1\}^n$ is uniformly distributed over $\{0,1\}^n_{ws + r_1}.$ In particular, setting $w = k'$ and $k'+q'$, we get the following.
	\begin{align*}
	\prob{\bm{a}\sim \{0,1\}^{n'}_{k'},\bm{P}}{\bm{P}(\bm{a}) \neq 0} = \prob{\bm{b}\sim \{0,1\}^{n}_{k}}{Q(\bm{b}) \neq 0} &\leq \varepsilon_0\\
	\prob{\bm{a}\sim \{0,1\}^{n'}_{k'+q'},\bm{P}}{\bm{P}(\bm{a}) \neq 0} = \prob{\bm{b}\sim \{0,1\}^{n}_{k+q}}{Q(\bm{b}) \neq 0} &\geq \varepsilon_1.	
	\end{align*}
	
    To find a suitable fixing of $\bm{P}$, we consider two cases.
	
	\begin{itemize}
	\item \textbf{Case 1: $e^{-\delta^2n/250 s\alpha} \geq 1/2$:} In this case, define two events $\mc{E}_0$ and $\mc{E}_1$ (depending only on the probabilistic polynomial $\bm{P}$) as follows.
	\[
	\mc{E}_0 := \prob{\bm{a}\sim \{0,1\}^{n'}_{k'}}{\bm{P}(\bm{a}) \neq 0} \geq 2\varepsilon_0, \text{ and } \ \ 
	\mc{E}_1 := \prob{\bm{a}\sim \{0,1\}^{n'}_{k'+q'}}{\bm{P}(\bm{a}) = 0}  \geq 2.5 \zeta_1
	\]
	where $\zeta_1 := 1-\varepsilon_1.$ Note that $\varepsilon_1 = e^{-\delta^2n/1000 s\alpha} \geq 2^{-0.25} > 0.8$ and hence $\zeta_1 < 0.2.$
	
	By Markov's inequality, with positive probability over the choice of $\bm{P}$, neither of the above events occurs. Fix such a polynomial $P$. Then, we have 
	\begin{align}
		 \prob{\bm{a}\sim \{0,1\}^{n'}_{k'}}{P(\bm{a}) \neq 0} &<  2\varepsilon_0 = \min\{2e^{-1000\delta^2n/s\alpha},2/2000\}\notag\\ 
        &\leq \min\{e^{-200\delta^2n/s\alpha}, 1/1000\},\label{eq:lem-ext-case1a}
	\end{align}
	where we used the simple fact that for any non-negative real number $\gamma$, we have the inequality $\min\{2\gamma,1/1000\}\leq \min\{\gamma^{0.2},1/1000\}$. We also have
	\begin{equation}
	\label{eq:lem-ext-case1b}
		 \prob{\bm{a}\sim \{0,1\}^{n'}_{k'+q'}}{P(\bm{a}) \neq 0} >  1-2.5\zeta_1  \geq (1-\zeta_1)^5 = \varepsilon_1^5 \geq e^{-\delta^2n/200s\alpha},
	\end{equation}
	where the second inequality uses the fact that $\zeta_1 \leq 0.2$ for any $\gamma\in [0,1],$ we have\footnote{This is a special case of the Boole-Bonferroni inequalities, which are closely related to the Principle of Inclusion-Exclusion.} $(1-\gamma)^5 \leq 1-5\gamma + 10\gamma^2$.
	 
	 \item \textbf{Case 2: $e^{-\delta^2n/250 s\alpha} < 1/2$:}  In this case, we proceed analogously, but define the `bad' events as follows.
	 \[
	\mc{E}_0' := \prob{\bm{a}\sim \{0,1\}^{n'}_{k'}}{\bm{P}(\bm{a}) \neq 0} \geq \frac{2\varepsilon_0}{\varepsilon_1}, \text{ and } \ \ 
	\mc{E}_1' := \prob{\bm{a}\sim \{0,1\}^{n'}_{k'+q'}}{\bm{P}(\bm{a}) \neq 0} < \frac{\varepsilon_1}{2}.
	\]
	By Markov's inequality, there is again a fixing $P$ of $\bm{P}$ such that neither of the above two events occurs. For such a polynomial $P$, we have
	\begin{equation}
	\label{eq:lem-ext-case2b}
		 \prob{\bm{a}\sim \{0,1\}^{n'}_{k'+q'}}{P(\bm{a}) \neq 0} \geq  \frac{\varepsilon_1}{2} \geq e^{-\delta^2n/1000s\alpha}\cdot e^{-\delta^2n/250s\alpha} = e^{-\delta^2n/200s\alpha},
	\end{equation}
	where the second inequality used our assumption that $e^{-\delta^2n/250 s\alpha} < 1/2$. We also have
	\begin{align}
		 \prob{\bm{a}\sim \{0,1\}^{n'}_{k'}}{P(\bm{a}) \neq 0} &<   \frac{\varepsilon_0}{\varepsilon_1/2} \leq e^{\delta^2n/200s\alpha}\cdot \varepsilon_0 \leq e^{-999\delta^2 n/s\alpha}\notag\\
		 &\leq e^{200\delta^2n/s\alpha} \leq \min\{e^{-200\delta^2n/s\alpha}, 1/1000 \}\label{eq:lem-ext-case2a},
	\end{align}
	where the second inequality used (\ref{eq:lem-ext-case2b}) above and the third and last inequalities use the fact that $e^{-\delta^2n/250 s\alpha} < 1/2$ to deduce that $e^{-1000\delta^2n/s\alpha} \leq 1/2000$ and $e^{-200\delta^2n/s\alpha} < 1/1000$ respectively.
	\end{itemize}
	
	Putting (\ref{eq:lem-ext-case1a}), (\ref{eq:lem-ext-case1b}), (\ref{eq:lem-ext-case2a}) and (\ref{eq:lem-ext-case2b}) together gives us that in both cases we have
	
	\begin{subequations}
	\label{eq:fixP}
	\begin{align}
	\prob{\bm{a}\sim \{0,1\}^{n'}_{k'}}{P(\bm{a}) \neq 0} &\leq \min\{e^{-200\delta^2n/s\alpha}, 1/1000\}	\label{eq:fixPa}\\
	\prob{\bm{a}\sim \{0,1\}^{n'}_{k'+q'}}{P(\bm{a}) \neq 0} &\geq e^{-\delta^2n/200s\alpha}.	\label{eq:fixPb} 
	\end{align}
	\end{subequations}
	
	To apply Lemma~\ref{lem:main} to $P$, we need to relate the above bounds to quantities defined in terms of $\delta' := q'/n'$ and $\alpha' := k'/n'.$ We claim that
	
	\begin{equation}
	\label{eq:lem-ext-ineqs}
		\frac{\delta^2n}{s\alpha} \leq \frac{(\delta')^2n'}{\alpha'} \leq \frac{2\delta^2n}{s\alpha}.
	\end{equation}
	 Assuming these inequalities, we observe that $P$ satisfies the hypotheses of Lemma~\ref{lem:main}. Applying this lemma gives us 
	\[
	\deg(Q) \geq \deg(\bm{P}) \geq \deg(P) = \Omega(q'),
	\]
	finishing the proof of Lemma~\ref{lem:main-extension}.
	
	It remains to prove (\ref{eq:lem-ext-ineqs}), which is a simple calculation.
	\begin{align*}
	 \frac{(\delta')^2n'}{\alpha'} &= \frac{(\delta'n')^2}{\alpha'n'} = \frac{(q')^2}{k'} = \frac{(q's)^2}{k's^2} \geq \frac{q^2}{sk} = \frac{\delta^2 n}{s\alpha},\\
	 \frac{(\delta')^2n'}{\alpha'} &= \frac{(q')^2}{k'} =\frac{(q's)^2}{k's^2} \leq \frac{q^2}{s(k-s)} = \frac{\delta^2 n}{s\alpha}\cdot \left(1-\frac{s}{k}\right)^{-1} \leq \frac{2\delta^2 n}{s\alpha}
	\end{align*}
	where the final inequality uses the fact that $\frac{s}{k}\leq \frac{q}{k} \leq 0.01.$
	\end{proof}
	
	\section{Applications}
	\label{sec:apps}
	
	\subsection{Tight Degree Lower Bounds for the Coin Problem}
	\label{sec:cp}
	
	We start with a definition.
	
	\begin{definition}[The $\delta$-Coin Problem] 
\label{def:coinproblem}
For any $\alpha\in [0,1]$ and integer $n\geq 1$, let $\mu_{\alpha}^n$ be the product distribution over $\{0,1\}^n$ obtained by setting each bit to $1$ independently with probability $\alpha.$ Let $\delta \in (0,1)$ be a parameter. 

Given a function $g: \{0,1\}^n \rightarrow \{0,1\}$, we say that \emph{$g$ solves the $\delta$-coin problem with error $\varepsilon$} if
\begin{equation}
\prob{\bm{x}\sim \mu_{(1/2)-\delta}^n}{g(\bm{x}) = 1} \leq \varepsilon \text{  and   } \prob{\bm{x}\sim \mu_{1/2}^n}{g(\bm{x}) = 1} \geq 1-\varepsilon.
\end{equation}
(This definition is sometimes~\cite{LSSTV-SICOMP} stated in terms of the distributions $\mu_{(1/2)-\delta}$ and $\mu_{(1/2)+\delta}$. This is essentially equivalent to the definition above.)
\end{definition}

Let $\F$ be a prime field of characteristic $p$, where $p$ is a fixed constant. We consider here the minimum degree of a polynomial $P\in \F[x_1,\ldots,x_n]$ that solves the $\delta$-coin problem with error $\varepsilon.$ 

By Lemma~\ref{lem:AW}, for any $n\geq 1$, there is a polynomial $P\in \F[x_1,\ldots,x_n]$ of degree $O(\delta n)$ that outputs $0$ on all inputs of weight $w\in (n((1/2)-3\delta/2),n(1/2-\delta/2))$ and $1$ on all inputs of weight $(n(1/2-\delta/2), n(1/2+\delta/2)).$ Using Lemma~\ref{lem:bernstein} (Bernstein's inequality), it can be easily checked that $P$ solves the $\delta$-coin problem with error $\varepsilon$ as long as $n \geq C\frac{1}{\delta^2}\log(1/\varepsilon)$ for some large enough constant $C > 0$. This yields a polynomial $P$ of degree $O(\frac{1}{\delta}\log(1/\varepsilon))$.

In earlier work~\cite{LSSTV-SICOMP}, we showed that this was tight for constant $\varepsilon.$ That is, we showed that any polynomial $P$ that solves the $\delta$-coin problem with error at most $1/10$ (say) must have degree $\Omega(1/\delta).$ This was also implied by an independent result of Chattopadhyay, Hosseini, Lovett and Tal~\cite{CHLT} (see~\cite{Agrawal}). Both proofs relied on slight strengthenings of Smolensky's~\cite{Smolensky87} lower bound on polynomials approximating the Majority function. It is not clear from these proofs, however, if this continues to be true for subconstant $\varepsilon.$ The main lemma (Lemma~\ref{lem:main}), or even its simpler version Lemma~\ref{lem:Dij}, shows that this is indeed true.

\begin{theorem}[Tight Degree Lower Bound for the $\delta$-coin problem for all errors]
\label{thm:coinproblem}
Assume $\F$ has characteristic $p$ and $\delta,\varepsilon$ are parameters going to $0$. Let $N\geq 1$ be any positive integer. Any polynomial $P\in \F[x_1,\ldots,x_N]$ that solves the $\delta$-coin problem with error $\varepsilon$ must have degree $\Omega(\frac{1}{\delta}\log(1/\varepsilon)).$
\end{theorem}

\begin{proof}
We assume that $\varepsilon$ is smaller than some small enough constant $\varepsilon_0$ (for larger $\varepsilon$, we can just appeal to the lower bound of~\cite{LSSTV-SICOMP}). 

Assume for now that $\delta = 1/k$ for some integer $k\geq 1$. Fix $n$ to be the least even integer such that $n \geq \frac{C}{\delta^2}\log(1/\varepsilon)$ for a large constant $C$  and $q := \delta n$ is a power of the characteristic~$p$. Note that $n\leq O(p)\cdot \frac{C}{\delta^2}\log(1/\varepsilon)= O(\frac{1}{\delta^2}\log(1/\varepsilon))$ as $p$ is a constant. Define the probabilistic polynomial $\bm{Q}\in \F[y_1,\ldots,y_n]$ obtained from $P$ by randomly replacing each variable of $P$ by a uniformly random variable among $y_1,\ldots,y_n$. For any $a\in \{0,1\}^n_{n/2}$, we have
\[
\prob{\bm{Q}}{\bm{Q}(a) = 0} = \prob{\bm{b} \sim \mu_{1/2}}{P(\bm{b}) = 0} \leq  \varepsilon,
\]
and similarly for $a\in \{0,1\}^n_{(n/2)-q}$, we have $\prob{\bm{Q}}{\bm{Q}(a) \neq 0}\leq  \varepsilon$. In particular, by Markov's inequality, there is a fixed polynomial $Q$ of degree at most $\deg(P)$ that satisfies
\[
\prob{\bm{a}\sim \{0,1\}^n_{n/2}}{Q(a) = 0} \leq  2\varepsilon \text{ and } \prob{\bm{a}\sim \{0,1\}^n_{(n/2)-q}}{Q(a) \neq 0} \leq 2\varepsilon.
\]
Hence, by Lemma~\ref{lem:Dij}, we have $\deg(P) = \Omega(\delta n) = \Omega(\frac{1}{\delta}\log(1/\varepsilon)).$

Now, if $\delta$ is not of the assumed form, we consider $k$ be the largest integer such that $\delta\leq 1/k$ and set $\delta' := 1/k$. Define $\alpha\in (0,1)$ by $\alpha = \delta/\delta'.$ Note that if $\bm{a},\bm{b}\in \{0,1\}$ are sampled independently from the distributions $\mu_{1/2-\delta'}^1$ and $\mu_{1/2-(\alpha/2)}^1$ respectively, then their parity $\bm{a}\oplus \bm{b}$ has the distribution $\mu_{1/2-\delta}^1.$ Now, if we define the probabilistic polynomial $\bm{R}(x_1,\ldots,x_n)$ by 
\[
\bm{R}(x_1,\ldots,x_n) = P(x_1\oplus \bm{y}_1,\ldots,x_n\oplus \bm{y}_n)
\]
where $\bm{y} = (\bm{y}_1,\ldots,\bm{y}_n)$ is sampled from $\mu_{1/2-(\alpha/2)}^n,$ then $\bm{R}$ solves the $\delta'$-coin problem with error at most $\varepsilon.$ Note also that $\deg(\bm{R}) \leq \deg(P)$ as for each fixed $\bm{y},$ each $x_i \oplus \bm{y}_i$ is a linear function of $x_i$.

Repeating the above argument with $\bm{R}$ instead of $P$ yields that $\deg(\bm{R}) = \Omega(\frac{1}{\delta'}\log(1/\varepsilon)) = \Omega(\frac{1}{\delta}\log(1/\varepsilon)).$ We thus get the same lower bound for $\deg(P).$
\end{proof}

\subsection{Tight Probabilistic Degree Lower bounds for Positive Characteristic}
\label{sec:pdeg}

We start with some basic notation and definitions and then state our result.

Throughout this section, let $\F$ be a field of fixed (i.e. independent of $n$) characteristic $p > 0$.  The main theorem of this section characterizes (up to constant factors) the $\varepsilon$-error probabilistic degree of every symmetric function and for almost all interesting values of $\varepsilon$.
	
	\begin{theorem}[Probabilistic Degree lower bounds over positive characteristic]
		\label{thm:pdeglbd}
		Let $n\in \mathbb{N}$ be a growing parameter. Let $f\in \sB_n$ be arbitrary and let $(g,h)$ be a standard decomposition of $f$ (see Section~\ref{sec:prelims} for the definition). Then for any $\varepsilon\in [1/2^n,1/3]$, we have
		 \[
		 \pdeg_\varepsilon^\F(f) = 
			\left\{
			\begin{array}{ll}
			\Omega(\sqrt{n\log(1/\varepsilon)}) & \text{if $\per(g) > 1$ and not a power of $p$,}\\
			\Omega(\min\{\sqrt{n\log(1/\varepsilon)},\per(g)\}) & \text{if $\per(g)$ a power of $p$ and $B(h) =0$,}\\
			\Omega(\min\{\sqrt{n\log(1/\varepsilon)},\per(g)  & \text{otherwise.}\\
			\ \ + \sqrt{B(h)\log(1/\varepsilon)} +\log(1/\varepsilon)\}) 
			\end{array}
			\right.
			\]
			Here the $\Omega(\cdot)$ notation hides constants depending on the characteristic $p$ of the field $\mathbb{F}.$
	\end{theorem}
	
	Note that this matches the upper bound construction from Theorem~\ref{thm:ubd-STV}.

	\subsubsection{Some Preliminaries}
	
	\begin{definition}[Restrictions]
	Given functions $f\in \sB_n$ and $g\in \sB_m$ where $m\leq n$, we say that $g$ is a restriction of $f$ if there is some $a\in [0,n-m]$ such that the identity
	\[
	g(x)=f(x1^a0^{n-m-a})
	\]
	holds for every $x\in \{0,1\}^n$. Or equivalently, that $g$ can be obtained from $f$ by setting some inputs to $0$ and $1$ respectively.\footnote{Note that exactly which inputs are set to $0$ or $1$ is not important, since we are dealing with \emph{symmetric} Boolean functions.}
	\end{definition}
	
	We will use the following obvious fact freely.
	\begin{observation}
	\label{obs:pdeg-restriction}
	If $g$ is a restriction of $f$, then for any $\delta > 0$, $\pdeg_{\delta}(g) \leq \pdeg_\delta(f).$
	\end{observation}
	
	In earlier work with Tripathi and Venkitesh~\cite{STV}, we showed the following near-optimal lower bound on the probabilistic degrees of Threshold functions.
		
	\begin{lemma}[Lemma 27 in~\cite{STV}]
	\label{lem:thr-lbd}
	 Assume $t \geq 1$. For any $\varepsilon\in [2^{-n},1/3],$
	\[
	\pdeg_{\varepsilon}(\Thr_n^t)= \Omega(\sqrt{\min\{t,n+1-t\}\log(1/\varepsilon)} + \log(1/\varepsilon)).
	\]
	\end{lemma}

	(The corresponding lemma in~\cite{STV} is only stated for $t\leq n/2$. However, as  $\Thr^{n+1-t}_n(x) = 1-\Thr^t_n(1-x_1,\ldots,1-x_n),$ the above lower bound holds for $t > n/2$ also.)
	
	The following classical results of Smolensky prove optimal lower bounds on the probabilistic degrees of some interesting classes of symmetric functions. 
	
	\begin{lemma}[Smolensky's lower bound for Majority function~\cite{Szegedy-thesis,Smolensky93}]
		\label{lem:majlbd}
		For any field $\F$, any $\varepsilon \in (1/2^n, 1/5),$ we have
		\[\pdeg_{\varepsilon}^\F(\Maj_n) = \Omega(\sqrt{n\log(1/\varepsilon)}).\]
	\end{lemma}
	
	\begin{lemma}[Smolensky's lower bound for MOD functions~\cite{Smolensky87}]
		\label{lem:modlbd}
		For $2\leq b\leq n/2$, any $\F$ such that $\charac(\F)$ is either zero or coprime to $b$, any $\varepsilon \in (1/2^n, 1/(3b))$, there exists an $i\in[0,b-1]$ such that
		\[\pdeg_{\varepsilon}^\F(\MOD^{b,i}_n) = \Omega(\sqrt{n\log(1/b\varepsilon)}).\]
	\end{lemma}
	
		We now show how to use our robust version of Heged\H{u}s's lemma to prove Theorem~\ref{thm:pdeglbd}. In fact, Lemma~\ref{lem:Dij} will suffice for this application.

	\subsubsection{Strategy and two simple examples}
	\label{sec:lbdchar+}
	
	The probabilistic degree lower bounds below will use the following corollary of Lemma~\ref{lem:Dij}.
	
	\begin{corollary}
	\label{cor:Dij}
		Let $n$ be a growing parameter and assume $\varepsilon \in [2^{-n/100},e^{-200}].$ Assume $t$ is an integer such that $t$ is a power of $p$ and furthermore, $t = \sqrt{n\ell}$ for some $\ell\in \mathbb{R}$ such that $100\leq \ell\leq \frac{1}{2}\cdot \ln(1/\varepsilon).$  Let $h\in \sB_n$ be any function such that $\spec h(\lfloor n/2\rfloor) \neq \spec h(\lfloor n/2\rfloor -t)$. Then, $\pdeg_{\varepsilon}(h) = \Omega(t).$
	\end{corollary}
	
	\begin{proof}
	By error reduction for probabilistic polynomials (Fact~\ref{fac:pdeg} item 1), it suffices to prove an $\Omega(t)$ lower bound on $\pdeg_{\varepsilon/2}(h).$
	
	Assume without loss of generality that $\spec h(\lfloor n/2\rfloor) =1$ and $\spec h(\lfloor n/2\rfloor -t)=0.$ Let $\bm{P}$ be an $(\varepsilon/2)$-error probabilistic polynomial for $h$. Then, we have
	\begin{align*}
	\prob{\bm{P},\bm{a}\sim \{0,1\}^n_{\lfloor n/2\rfloor}}{\bm{P}(\bm{a}) \neq 1} &\leq \varepsilon/2\\
	\prob{\bm{P},\bm{b}\sim \{0,1\}^n_{\lfloor n/2\rfloor-t}}{\bm{P}(\bm{b}) \neq 0} &\leq \varepsilon/2
	\end{align*}
	
	Thus, we have 
	\[
	\avg{\bm{P}}{\prob{\bm{a}\sim \{0,1\}^n_{\lfloor n/2\rfloor}}{\bm{P}(\bm{a}) \neq 1}  + \prob{\bm{b}\sim \{0,1\}^n_{\lfloor n/2\rfloor-t}}{\bm{P}(\bm{b}) \neq 0}}\leq \varepsilon,
	\]
	and hence, by averaging, there is a polynomial $P$ in the support of the distribution of $\bm{P}$ such that 
	\[
	\prob{\bm{a}\sim \{0,1\}^n_{\lfloor n/2\rfloor}}{P(\bm{a}) \neq 1}  + \prob{\bm{b}\sim \{0,1\}^n_{\lfloor n/2\rfloor-t}}{P(\bm{b}) \neq 0} \leq \varepsilon.
	\]
	Applying Lemma~\ref{lem:Dij} to $P$ yields
	\[
	\deg(\bm{P}) \geq \deg(P)  = \Omega(t).\qedhere
	\]
	\end{proof}
	
	To illustrate the usefulness of Corollary~\ref{cor:Dij}, we prove optimal lower bounds on the probabilistic degrees for two interesting classes of functions (both of which will be subsumed by Theorem~\ref{thm:pdeglbd}). 
	
	\begin{corollary}
	\label{cor:lbd-modq}
	Let $\varepsilon \in (0,1/3]$ be a constant. Let $q$ be any integer relatively prime to $p$ such that $q\leq 0.99n.$ Then the $\varepsilon$-error probabilistic degrees of $\EThr_n^{\lfloor n/2 \rfloor}$ and $\MOD_n^q$ are $\Omega(\sqrt{n}).$
	\end{corollary}
	
	Known lower bounds (Lemmas~\ref{lem:majlbd} and~\ref{lem:modlbd}) can be used to prove similar lower bounds to the one given above, but with additional log-factor losses~(see Lemma~\ref{lem:modlbd}, which requires the error to be subconstant, and \cite{STV}). However, we do not know how to prove the above tight (up to constants) lower bound without appealing to Lemma~\ref{lem:Dij}. In particular, we do not know how to prove the above in characteristic $0$. 
	
	\begin{proof}
	We use Corollary~\ref{cor:Dij}. We will use $\EThr_n^{\lfloor n/2\rfloor}$ and $\MOD_n^q$ to construct functions that distinguish between weights $\lfloor n/2\rfloor$ and $\lfloor n/2\rfloor-t$ for suitable $t = \Omega(\sqrt{n}).$ Corollary~\ref{cor:Dij} then implies the required lower bound.
	
	For $h=\EThr_n^{\lfloor n/2\rfloor}$, note that $\spec h(\lfloor n/2\rfloor) \neq \spec h(\lfloor n/2\rfloor - t)$ for any $t < \lfloor n/2\rfloor$. In particular, setting $t$ to be the smallest power of $p$ such that $t \geq \sqrt{100 n}$ and $\varepsilon_0 = e^{-2t^2/n},$ we get by Corollary~\ref{cor:Dij} that $\pdeg_{\varepsilon_0}(h) = \Omega(t) = \Omega(\sqrt{n}).$ By error-reduction for probabilistic polynomials (Fact~\ref{fac:pdeg} item 1), we also have the same lower bound (up to constant factors) for any $\varepsilon \leq 1/3.$ This proves the claim in the case that $h = \EThr_n^{\lfloor n/2\rfloor}.$
	
	For $h = \MOD_n^q$, we make some minor modifications to the above idea. Let $r\in [0,q-1]$ be such that $r + \lfloor (n-q)/2\rfloor \equiv 0\pmod{q}.$ Define $h'\in \sB_{n-q}$ by 
	\[
	h'(x) = h(x1^r0^{q-r}).
	\]
	Set $t$ to be the smallest power of $p$ such that $t \geq \sqrt{100 (n-q)}$ and $\varepsilon_0 = e^{-2t^2/(n-q)}.$ Note that $\spec h'(\lfloor (n-q)/2\rfloor) = \spec h(r+\lfloor (n-q)/2\rfloor) = 1$ as $r + \lfloor (n-q)/2\rfloor \equiv 0\pmod{q}.$ On the other hand,  $r + \lfloor (n-q)/2\rfloor -t \not\equiv 0\pmod{q}$ as $t$ is a power of $p$ and hence not divisible by $q$, which implies that $\spec h'(\lfloor (n-q)/2\rfloor -t ) = 0$. Thus, by Corollary~\ref{cor:Dij}, we get $\pdeg_{\varepsilon_0}(h') = \Omega(t) = \Omega(\sqrt{n}).$ 
	\end{proof}

	\subsubsection{Proof of Theorem~\ref{thm:pdeglbd}}
	
	The proof of this theorem closely follows our probabilistic degree lower bounds in~\cite{STV} with careful modifications to avoid the log-factor losses therein.
	
	Let $f\in \sB_n$ be arbitrary and let $(g,h)$ be a standard decomposition of $f$.

	We start with a lemma that proves lower bounds on $\pdeg_\varepsilon(f)$ as long as $\per(g)$ is large.
	
	\begin{lemma}
	\label{lem:aperiodic}
	Fix any $\varepsilon\in [2^{-n},1/3].$ Assume that $f$ is such that $\per(g) > \sqrt{n\log(1/\varepsilon)}.$ Then
	\[
	\pdeg_\varepsilon(f) = \Omega(\sqrt{n\log(1/\varepsilon)}).
	\]
	\end{lemma}
	
	\begin{proof}
	We first prove the lemma under the assumption that $\varepsilon\in [2^{-n/1000},e^{-10000p^2}].$  	
	
	Fix $m$ to be the largest power of $p$ upper bounded by $ \frac{1}{4}\sqrt{n\log(1/\varepsilon)}$. 
	
	Since $\per(g) > \sqrt{n\log(1/\varepsilon)} \geq m,$ there is no function $g'\in \sB_n$ that has period $m$ and agrees with $f$ on the interval $I := [\lceil n/3\rceil+1, \lfloor 2n/3\rfloor].$ Thus, there exists some $r\in I$ such that $r + m \in I$ and $\spec f(r) \neq \spec f(r + m).$ 
	
	Let $k = \lceil n/2\rceil$. Note that $r \geq \lceil n/3\rceil \geq  k/2$ and $r + m\leq \lfloor 2n/3\rfloor.$ Define $F\in \sB_k$ by setting 
	\[
	F(x) = f(x1^{a}0^{b})
	\]
	where $a= r+m-\lfloor k/2\rfloor$ and $b = n-k-a$ (it can be checked that $a,b$ are non-negative for parameters $r,m,k$ as above). Note that $\spec F(\lfloor k/2\rfloor) = \spec f(\lfloor k/2\rfloor + a) = \spec f(r+m)$ and similarly that $\spec F(\lfloor k/2\rfloor - m) = \spec f(r).$ We thus obtain $\spec F(\lfloor k/2\rfloor) \neq \spec F(\lfloor k/2\rfloor - m).$

	Note that by the bounds on $\varepsilon$ assumed above
	\begin{equation}
	\label{eq:char+aperiodic}
	m\geq \frac{1}{4p} \sqrt{n\log(1/\varepsilon)}\geq 20\sqrt{n}.
	\end{equation}	
	 Using Corollary~\ref{cor:Dij}, we hence get 
	\[
	\pdeg_{\varepsilon}(f) \geq \pdeg_{\varepsilon/2}(F) = \Omega(m) = \Omega(\sqrt{n\log(1/\varepsilon)})
	\]
	which proves the lemma under the assumption on $\varepsilon$ above. (We use the bounds on $\varepsilon$ to ensure that $2^{-k/200}\leq \varepsilon\leq e^{-2m^2/k},$ which is part of the hypothesis of Corollary~\ref{cor:Dij}.)
	
	If $\varepsilon \in [2^{-n},2^{-n/10000p^2}],$ then for $\varepsilon_0 = 2^{-n/10000p^2},$ we have
	\[
	\pdeg_{\varepsilon}(f) \geq \pdeg_{\varepsilon_0}(f) = \Omega(\sqrt{n\log(1/\varepsilon_0)}) = \Omega(\sqrt{n\log(1/\varepsilon)})
	\]
	which implies the desired lower bound.\footnote{Note that we assume that the characteristic is a fixed positive constant and hence the $\Omega(\cdot)$ can hide constants depending on $p$.}
	
	On the other hand, if $\varepsilon > e^{-10000p^2}$, we proceed as follows. We construct $F$ as above, but we may no longer have $m\geq 20\sqrt{n}$ as implied by (\ref{eq:char+aperiodic}). However, for $F'\in \sB_{k'}$ defined by
	\[
	F'(x) = F(x0^t1^t)
	\]
	for suitably chosen $t\leq k/2$, we can ensure that $m \in [10\sqrt{k'},20\sqrt{k'}].$ Note that $\spec F'(\lfloor k'/2\rfloor) = \spec F(\lfloor k/2\rfloor)$ and $\spec F'(\lfloor k'/2\rfloor - m) = \spec F(\lfloor k/2\rfloor - m)$. Hence, for $\varepsilon_1 = e^{-10000}$, Corollary~\ref{cor:Dij}  implies 
	\[
	\pdeg_{\varepsilon_1}(f) \geq \pdeg_{\varepsilon_1}(F') = \Omega(m) = \Omega(\sqrt{n\log(1/\varepsilon_1)}).
	\]
	By error reduction (Fact~\ref{fac:pdeg} item 1), the same lower bound holds for $\pdeg_\varepsilon(f)$ as well.
	\end{proof}
	
	The next lemma allows us to prove a weak lower bound on $\pdeg_\varepsilon(f)$ depending only on its periodic part $g$.	
	
	\begin{lemma}\label{lem:preg-f-weak-char+}
	For any $\varepsilon\in [2^{-n},1/3],$
		\[
		\pdeg_\varepsilon(f)\ge\begin{cases}
		\Omega(\sqrt{n\log(1/\varepsilon)}),&\tx{if }\per(g)\tx{ is not a power of }p\\
		\Omega(\min\{\per(g),\sqrt{n\log(1/\varepsilon)}\}),&\tx{if }\per(g)\tx{ is a power of }p.
		\end{cases}
		\]
	\end{lemma}
	
	\begin{proof}	
	By Fact~\ref{fac:pdeg} item 1 (error reduction), we know that $\pdeg_\varepsilon(g) = \Theta(\pdeg_\delta(g))$ as long as $\delta = \varepsilon^{\Theta(1)}.$ In particular, we may assume without loss of generality that $\varepsilon\in [2^{-n/10000},e^{-10000p^2}].$
	
		Let $b := \per(g)$. If $\per(g) > \sqrt{n\log(1/\varepsilon)}$, we are done by Lemma~\ref{lem:aperiodic}. So we assume that $b\leq \sqrt{n\log(1/\varepsilon)}.$ In particular, this implies that $b \leq n/100.$
	
	We have two cases.
	
		\subparagraph*{\(b\) is not a power of \(p\).}  Let $n_1$ be the largest power of $p$ upper bounded by $\frac{1}{4}\sqrt{n\log(1/\varepsilon)}.$ By the constraints on $\varepsilon$, we have $10\sqrt{n}\leq n_1 \leq n/100.$ 
		
		Let $b_1\in [0,b-1]$ such that $b_1\equiv n_1 \pmod{b}$; note that $b_1\neq 0$ as $b$ is not a power of $p$. As $b_1$ is smaller than $b = \per(g)$, there must exist $r \in [0, n-b_1]$ such that 
		\[
		\spec g(r) \neq \spec g(r+b_1).
		\]
		Assume that we choose the smallest $r \geq n/2$ so that this condition holds. Then we have $r\leq n/2 + b \leq 51\cdot n/100.$ Fix this $r$. As $\spec g(r) \neq \spec g(r+b_1),$ we also have $\spec g(r) \neq \spec g(r+b_1+k\cdot b)$ for any integer $k$ such that $0\leq r+b_1+kb\leq n.$ In particular, as $b_1\equiv n_1\pmod{b},$ we note that $\spec g(r) \neq \spec g(r+n_1).$ As $n_1\leq n/100,$ we have
		\[
		n/2 \leq r\leq r+n_1 \leq n/2 + n/50.
		\]
		As $\spec g(i) = \spec f(i)$ for all $i\in [\lceil n/3\rceil + 1, \lfloor 2n/3\rfloor ]$, we have $\spec f(r) \neq \spec f(r+n_1).$ Without loss of generality, we assume that $\spec f(r) = 0$ and $\spec f(r+n_1) = 1.$
		
		Let $m = \lceil n/2\rceil.$ We define $F\in \sB_m$ as follows.
		\[
		F(x) = f(x1^a0^{n-m-a})
		\]
		where $a$ is chosen so that $\spec F(\lfloor m/2\rfloor) = \spec f(r+n_1) = 1$. This also has the consequence that $\spec F(\lfloor m/2\rfloor - n_1) = \spec f(r) = 0.$ By Corollary~\ref{cor:Dij}, we get $\pdeg_{\varepsilon}(F) = \Omega(n_1) = \Omega(\sqrt{n\log(1/\varepsilon)}),$ proving the lemma in this case.

		\subparagraph*{\(b\) is a power of \(p\).} In this case, we first choose parameters $m,\delta$ with the following properties.
		\begin{enumerate}
		\item[(P1)] $m\in [n]$ with $m\geq 20b$ and $m\equiv n \pmod{2}.$
		\item[(P2)] $1/3\geq \delta \geq \max\{\varepsilon,1/2^{m}\}.$
		\item[(P3)] $\sqrt{m\log(1/\delta)} < b.$
		\item[(P4)] $\sqrt{m\log(1/\delta)} = \Omega(\min\{b,\sqrt{n\log(1/\varepsilon)}\}) = \Omega(b).$  (Recall that $b\leq \sqrt{n\log(1/\varepsilon)}$.)
		\end{enumerate}
		
		We will show below how to find $m,\delta$ satisfying these properties. Assuming this for now, we first prove the lower bound on $\pdeg_\varepsilon(f).$  
		
		Define $F\in \sB_m$ as follows. 
		\[
		F(x) = f(x0^{t}1^{t})
		\]
		for $t = (n-m)/2$. We observe that if $(G,H)$ is a standard decomposition of $F$, then $\per(G) \geq b.$ To see this, note that by Corollary~\ref{cor:string}, we have 
		\[
		\spec g|_{[\lfloor n/2\rfloor, \lfloor n/2\rfloor + b-1 ]} \neq \spec g|_{[\lfloor n/2\rfloor + i, \lfloor n/2\rfloor + i + b-1 ]}
		\]
		for any $i\in [b-1]$. As $f$ and $g$ agree on inputs of weight from $[\lfloor n/3\rfloor + 1,\lfloor 2n/3\rfloor],$ the same non-equality holds for $\spec f$ also. Further, as $\spec F(\lfloor m/2\rfloor + j) = \spec f(\lfloor n/2\rfloor + j)$ for $j\leq m/2$, we also get
		\[
		\spec F|_{[\lfloor m/2\rfloor, \lfloor m/2\rfloor + b-1 ]} \neq \spec F|_{[\lfloor m/2\rfloor + i, \lfloor m/2\rfloor + i + b-1 ]}.
		\]
		for any $i\in [b-1]$ (we have used here the fact that $m\geq 20b$ which holds by (P1)). Finally, as $F$ and $G$ agree on inputs of weight from $[\lfloor m/3\rfloor+1, \lfloor 2m/3\rfloor]\supseteq [\lfloor m/2\rfloor, \lfloor m/2\rfloor + 2b],$ the above non-equality holds for $G$ as well. This implies that $G$ cannot have period smaller than $b$.
		
		By (P3), we have $\per(G) > \sqrt{m\log(1/\delta)}.$ Lemma~\ref{lem:aperiodic} above and (P4) now imply that $\pdeg_\delta(F) = \Omega(\min\{b,\sqrt{n\log(1/\varepsilon)}\}).$ However, as $\delta\geq \varepsilon$ (by (P2)) and $F$ is a restriction of $f$, the same lower bound holds for $\pdeg_\varepsilon(f)$ as well. This proves the lemma modulo the existence of $m,\delta$ as above. We justify this now. 
		
		\begin{enumerate}
		\item If $b\leq 10\sqrt{n},$ we take $m$ to be the largest integer such that $m\equiv n \pmod{2}$ and $m\leq b^2/100.$ The parameter $\delta$ is set to $1/3.$ 
		\item If $10\sqrt{n} < b \leq n/100,$ then we take $m$ to be the largest integer such that $m\equiv n \pmod{2}$ and $m\leq n/2.$   The parameter $\delta = \max\{\varepsilon,2^{-b^2/2m}\}.$
		\end{enumerate}
		Note that as observed above, we have $b \leq n/100,$ and hence, the above analysis subsumes all cases.
		
		In each case, the verification of properties (P1)-(P4) is a routine computation. (We assume here that $b$ is greater than a suitably large constant, since otherwise the statement of the lemma is trivial.) This concludes the proof.
\end{proof}		

	We now prove a lower bound on $\pdeg_\varepsilon(h).$
	
	\begin{lemma}
	\label{lem:preg-h-char+}
	Assume $B(h)\geq 1$. Then, $\varepsilon\in [2^{-n},1/3],$
		\[
		\pdeg_\varepsilon(h)=\Omega(\sqrt{B(h)\log(1/\varepsilon)}+\log(1/\varepsilon)).
		\]
	\end{lemma}	
	
	\begin{proof}
	Similar to the proof of Lemma~\ref{lem:preg-f-weak-char+}, we may assume without loss of generality that $\varepsilon\in [2^{-n/10000},e^{-10000p^2}].$ 
	
	Let $B(h) = b.$ Recall (Observation~\ref{obs:decomp}) that $B(h) \leq \lceil n/3\rceil.$ Further, by definition of $B(h),$ we have either $\spec h(b-1)=1$ or $\spec h(n-b+1)=1.$ We assume that $\spec h(n-b+1) = 1$ (the other case is similar).
	
	The lemma is equivalent to showing that $\pdeg_{\varepsilon}(h) = \Omega(\max\{\sqrt{B(h)\log(1/\varepsilon)},\log(1/\varepsilon)\})$. We do this based on a case analysis based on the relative magnitudes of $\log(1/\varepsilon)$ and $b.$
	
	Assume for now that $\varepsilon \leq 2^{-b/1000}$. In this case, we show a lower bound of $\Omega(\log(1/\varepsilon)).$ To see this, set $m = \lceil n/4\rceil$ and consider the restriction $H\in \sB_{m}$ obtained as follows.
	\[
	H(x) = h(x1^{n-b+1-m}0^{b-1}).
	\]
	Note that as $\spec h$ is the constant $0$ function on the interval $[b,n-b],$ the function $H$ is computing the AND function on $m$ inputs. By Lemma~\ref{lem:thr-lbd}, we immediately have $\pdeg_{\varepsilon}(h) \geq \pdeg_\varepsilon(H) = \Omega(\log(1/\varepsilon))$ proving the lemma in this case.
	
	Now assume that $\varepsilon > 2^{-b/1000}.$ In this case, we need to show that $\pdeg_\varepsilon(h)$ is lower bounded by  $\Omega(\sqrt{b\log(1/\varepsilon)}).$ To prove this, consider the restriction $H\in \sB_{2b-2}$ defined by $H(x) = h(x1^{n-2b+2}).$ Since $\spec h$ is the constant $0$ function on the interval $[b,n-b]$ and $\spec h(n-b+1) = 1,$ it follows that the periodic part of $H$ has period $\Omega(b)$. It then follows from Lemma~\ref{lem:aperiodic} that $\pdeg_\varepsilon(h) = \Omega(\sqrt{b\log(1/\varepsilon)}).$ This concludes the proof of the lemma.
	\end{proof}
	
	Now, we are ready to prove Theorem~\ref{thm:pdeglbd}.
	
	\begin{proof}[Proof of Theorem~\ref{thm:pdeglbd}]
	By Lemma~\ref{lem:preg-f-weak-char+}, we already have the desired lower bound on $\pdeg_\varepsilon(f)$ in any of the following scenarios.
	
	\begin{itemize}
	\item $\per(g)$ is not a power of $p$, or
	\item $\per(g)$ is a power of $p$ and $\per(g) \geq \sqrt{n\log(1/\varepsilon)},$ or
	\item $B(h) = 0.$
	\end{itemize}
	
	So from now, we assume that $\per(g)$ is a power of $p$ upper-bounded by $\sqrt{n\log(1/\varepsilon)}$ and that $B(h)\geq 1.$ In this case, Lemma~\ref{lem:preg-f-weak-char+} shows that $\pdeg(f) = \Omega(\per(g)).$ On the other hand, since $B(h) \leq n$ and $\varepsilon\geq 2^{-n},$ the lower bound we need to show is $\Omega(\per(g) + \sqrt{B(h)\log(1/\varepsilon)}+\log(1/\varepsilon)).$ By Lemma~\ref{lem:preg-h-char+}, it suffices to show a lower bound of $\Omega(\per(g)+\pdeg_\varepsilon(h)).$
	
	The analysis splits into two simple cases.
	
	Assume first that $\pdeg_{\varepsilon}(h) \leq 4\cdot\per(g)$. In this case, we are trivially done, because we already have $\pdeg(f) = \Omega(\per(g))$, which is $\Omega(\pdeg(g) + \pdeg_\varepsilon(h))$ as a result of our assumption.
	
	Now assume that $\pdeg_\varepsilon(h) > 4\cdot \per(g).$ We know that $f = g\oplus h$ and hence $h = f\oplus g.$ Hence, we have
	\begin{align*}
	\pdeg_\varepsilon(h) \leq 2(\pdeg_{\varepsilon/2}(f) + \pdeg_{\varepsilon/2}(g)) \leq  O(\pdeg_{\varepsilon}(f)) + 2 \per(g),
	\end{align*}
	where the first inequality is a consequence of Fact~\ref{fac:pdeg} item 2 and the second follows from error-reduction and Theorem~\ref{thm:ubd-STV}. The above yields
	\[
	\pdeg_{\varepsilon}(f) = \Omega( (\pdeg_\varepsilon(h) - 2\cdot \per(g)) ) = \Omega(\pdeg_\varepsilon(h))  = \Omega(\per(g) + \pdeg_\varepsilon(h)).
	\]
	This finishes the proof.
	\end{proof}
	
	\subsection{A Robust Version of Galvin's Problem}
	\label{sec:galvin}

	We recall here a combinatorial theorem of Heged\H{u}s~\cite{Hegedus} regarding set systems. The theorem (and also our robust generalization given below) is  easier to prove in the language of indicator vectors, so we state it in this language.
	
	Given any vectors $u,v\in \F^n$ for any field $\F,$ we define $\ip{u}{v} := \sum_{j\in [n]} u_j v_j.$
	
	\begin{theorem}
	\label{thm:Hegedus}
	Assume $n = 4p$, for a  large enough prime $p$. Let $u^{(1)},\ldots,u^{(m)}\in \{0,1\}^n_{n/2}\subseteq \mathbb{Z}^n$ be such that for each $v\in \{0,1\}^n_{n/2}$, there is an $i\in [m]$ such that $\ip{u^{(i)}}{v} = p.$ Then $m \geq p.$
	\end{theorem}
	
	The above theorem is nearly tight as can be seen by taking the indicator vectors of the sets $S_i = \{i,(i+1),\ldots,i+(n/2)-1\}$ for $i\in [n/2]$. Improvements on the above theorem (some of them asymptotically tight) were proved recently by Alon et al.~\cite{AKV} and Hrube\v{s} et al.~\cite{HRRY}.
	
Using the robust version of Heged\H{u}s's lemma, we can prove tight robust versions of the above statement.

	\begin{remark}
	\label{rem:robustcaveat}
We can prove a robust generalization (stated below) in a slightly more general setting where the $i$th inner product $\ip{u^{(i)}}{v}$ is supposed to take a value $b_i$ (which is not necessarily~$p$). Similar to Theorem~\ref{thm:Hegedus} above, it is easy to note that our robust version is tight up to constant factors. 

However, if we consider the robust version of the original statement of Theorem~\ref{thm:Hegedus} (where all the inner products take value $p$), then while our lower bound continues to hold, it is not clear whether it is tight (except in the settings where $\varepsilon$ is either a constant or $2^{-\Omega(n)}$). We conjecture that it is. %Whether this is true may be an interesting question.
	\end{remark}

	We now prove a robust version of Theorem~\ref{thm:Hegedus}.
	
	\begin{theorem}
	\label{thm:robustHegedus}
	Assume $n$ is a growing even integer parameter and $\varepsilon \in [2^{-n}, 1/2].$. Let $u^{(1)},\ldots,u^{(m)}\in \{0,1\}^n_{n/2}\subseteq \mathbb{Z}^n$ and $b_1,\ldots,b_m\leq n$ be such that 
	\[
	\prob{\bm{v}\sim \{0,1\}^n_{n/2}}{\text{$\exists i\in [m]$ s.t. $\ip{u^{(i)}}{\bm{v}} = b_i$}} \geq 1-\varepsilon.
	\]
	 Then $m = \Omega(\sqrt{n\log(1/\varepsilon)}).$
	\end{theorem}
	
	The theorem can easily seen to be tight up to constant factors. For $t = C\cdot \sqrt{n\log(1/\varepsilon)}$, set $m = 2t+1$ and take $u^{(1)} = u^{(2)} = \cdots = u^{(m)} = 1^{n/2}0^{n/2}$ and $b_1 = (n/4) -t,b_2 = (n/4)-t+1,\ldots,b_m = (n/4)+ t.$ By standard Chernoff bounds for the Hypergeometric distribution, we immediately get that this set of hyperplanes satisfy the above condition for a large enough choice of the constant $C$.
	
	We need the following standard bound on binomial coefficients. For completeness, we include the proof in Appendix~\ref{appsec:hypergeometric}.

   	\begin{restatable}{claim}{clmhypergeometric}
	\label{clm:hypergeometric}
	Let $n$ be an even integer and $m$ a non-negative integer with $m\leq n/2$. Then, for any $k,\ell\in \{0,\ldots,\lfloor m/2\rfloor\}$ with $\ell \leq k$, we have
	\[
	\frac{\binom{n/2}{\lfloor m/2\rfloor - k}\binom{n/2}{\lceil m/2\rceil + k}}{\binom{n/2}{\lfloor m/2\rfloor-\ell}\binom{n/2}{\lceil m/2\rceil+\ell}}\leq \exp(-\Omega((k^2-\ell^2)/m)).
	\]
	\end{restatable}
	
	Given the above, we can prove Theorem~\ref{thm:robustHegedus} as follows.
	
	\begin{proof}[Proof of Theorem~\ref{thm:robustHegedus}]
	Recall that for any fixed $u\in \{0,1\}^n_{n/2}$ and any $b\in \mathbb{Z}$, the probability that a uniformly random $v\in \{0,1\}^n$ satisfies $\ip{u}{v} = b$ is at most $O(1/\sqrt{n})$. In particular, we must have $m = \Omega(\sqrt{n})$ for any $\varepsilon \leq 1/2.$ This proves the result for $\varepsilon = \Omega(1).$
	
	Hence, we may assume that $\varepsilon$ is smaller than any fixed constant. We can also assume that $\varepsilon \geq 2^{-\delta n}$ for a small enough constant $\delta$. Assume that $m\leq \sqrt{n\log(1/\varepsilon)}$.
	
	We call $i\in [m]$ \emph{balanced} if $|b_i- \frac{n}{4}|\leq t$ where $ t := C\sqrt{n\log(1/\varepsilon)}$ for a large enough constant~$C$. If $i$ is not balanced, then we have for a uniformly random $\bm{v}\sim \{0,1\}^n_{n/2}$,
	\[
	\prob{\bm{v}}{\ip{u^{(i)}}{\bm{v}} = b_i} \leq \frac{\binom{n/2}{n/4+t}\binom{n/2}{n/4-t}}{\binom{n}{n/2}}\leq \exp(-\Omega(t^2/n))\frac{\binom{n/2}{n/4}^2}{\binom{n}{n/2}} < \frac{\varepsilon^2}{\sqrt{n}}.
	\]
 The second inequality above follows from Claim~\ref{clm:hypergeometric}, and the third follows from the Stirling approximation and using the fact that $C$ is a large enough constant. In particular, if $B$ is the set of balanced $i$, we have 
	\[
	\prob{\bm{v}}{\exists i\not\in B,\ \ip{u^{(i)}}{\bm{v}} = b_i} \leq m\cdot \frac{\varepsilon^2}{\sqrt{n}} < \varepsilon
	\]
	where we used the fact that $m\leq \sqrt{n\log(1/\varepsilon)}$. We can thus consider only $\{u^{(i)}\ |\ i\in B\}$, which satisfy the hypothesis with error probability $\varepsilon_1 := 2\varepsilon.$ 
	
	Now consider the polynomial 
	\[
	P(x_1,\ldots,x_n) = \prod_{i\in B}(\ip{u^{(i)}}{x} - b_i).
	\]
	We know that $P$ vanishes at a random point of $\{0,1\}^n_{n/2}$ with probability at least $1-\varepsilon_1.$ Now, fix any prime $p\in [10t,20t]$ (such a prime exists by standard number-theoretic results). We claim that for any $i\in B$ and a uniformly random point $\bm{v}\in \{0,1\}^n_{n/2-p}$, we have  
	\begin{equation}
	\label{eq:robustHeg}
	\prob{\bm{v}}{\ip{u^{(i)}}{\bm{v}}\equiv b_i\pmod{p}} \leq \frac{\varepsilon^2}{\sqrt{n}} %= \sum_{j\in \mathbb{Z}} \prob{\bm{v}}{\ip{u^{(i)}}{\bm{v}}= b_i + jp}
	\end{equation}
	 for a large enough constant $C$. Informally speaking, the reason for this inequality is as follows: the expected value of $\ip{u^{(i)}}{\bm{v}}$ is $(n/4)-p/2$ and any number $b \equiv b_i\pmod{p}$ is far from this expectation. To prove this, let $s = n/2-p$. Note that $s = \Omega(n)$ as long as $t$ is small enough in relation to $n$, which happens if $\delta$ is assumed to be a small enough constant. Using the fact that $i$ is balanced, we note that 
	 \begin{align*}
			 \Delta_i := b_i - \ceil*{\frac{s}{2}} &\leq \frac{n}{4} + t - \left(\frac{n}{4} - \ceil*{\frac{p}{2}}\right)	\leq \frac{2p}{3}\\
			 \Delta_i &\geq \frac{n}{4} -t - \left(\frac{n}{4} - \ceil*{\frac{p}{2}}\right) \geq \frac{p}{3}.
	 \end{align*}
	  We thus have
	 \begin{align*}
		\prob{\bm{v}}{\ip{u^{(i)}}{\bm{v}}\equiv b_i\pmod{p} \wedge \ip{u^{(i)}}{\bm{v}}\geq \ceil*{\frac{s}{2}}} &= \sum_{j \geq 0} \prob{\bm{v}}{\ip{u^{(i)}}{\bm{v}}= b_i +  jp}\\
		&= \sum_{j \geq 0} \prob{\bm{v}}{\ip{u^{(i)}}{\bm{v}} - \ceil*{\frac{s}{2}}= \Delta_i +  jp}\\
		& = \sum_{j \geq 0}\frac{\binom{n/2}{\ceil*{\frac{s}{2}} + \Delta_i + jp}\binom{n/2}{\floor*{\frac{s}{2}}- \Delta_i- jp}}{\binom{n}{s}}\\
	\text{(by Claim~\ref{clm:hypergeometric})}	&= \frac{\binom{n/2}{\ceil*{\frac{s}{2}}} \binom{n/2}{\floor*{\frac{s}{2}}}}{\binom{n}{s}}\cdot  \sum_{j \geq 0} \exp(-\Omega((\Delta_i+jp)^2)/s))\\
	\text{(Stirling approximation and $s  = \Omega(n)$}&\leq O\left(\frac{1}{\sqrt{n}}\right) \cdot \sum_{j\geq 0} \exp(-\Omega(\Delta_i^2 + j p^2)/s)\\
	\text{($p^2/s \geq C^2$ )} &\leq O\left(\frac{1}{\sqrt{n}}\right) \cdot \exp(-\Omega(\Delta_i^2/s)) \cdot \sum_{j\geq 0} \exp(-\Omega(C^2 j)) \\
	\text{(for large enough $C$ )} &\leq O\left(\frac{1}{\sqrt{n}}\right) \cdot \exp(-\Omega(\Delta_i^2/s))\cdot 2\\
	&= O\left(\frac{1}{\sqrt{n}}\right) \cdot \exp(-\Omega(p^2/s)).
	 \end{align*}
	 In a similar way, we also get
	 \begin{align*}
		\prob{\bm{v}}{\ip{u^{(i)}}{\bm{v}}\equiv b_i\pmod{p} \wedge \ip{u^{(i)}}{\bm{v}}\leq \floor*{\frac{s}{2}}} \leq O\left(\frac{1}{\sqrt{n}}\right) \cdot \exp(-\Omega(p^2/s)).
	 \end{align*}
	 Overall, we thus obtain for any $i\in B$,
	 \begin{align*}
		\prob{\bm{v}}{\ip{u^{(i)}}{\bm{v}}\equiv b_i\pmod{p}} \leq O\left(\frac{1}{\sqrt{n}}\right) \cdot \exp(-\Omega(p^2/s))\leq \frac{\varepsilon^2}{\sqrt{n}}
	 \end{align*}
	 as long as $C$ is a large enough constant. Union bounding over the at most $m\leq \sqrt{n\log(1/\varepsilon)}$ elements of $B$, we see that
	\[
	\prob{\bm{v}\in \{0,1\}^n_{n/2-p}}{P(\bm{v}) \equiv 0 \pmod{p}} \leq \varepsilon.
	\]
	From now on, we consider the polynomial $P$ as an element of $\F_p[x_1,\ldots,x_n].$ At this point, we would like to apply Lemma~\ref{lem:main} to the polynomial $P$ and finish the proof. Unfortunately, the error parameter $\varepsilon_1$ above is not small enough to apply Lemma~\ref{lem:main} directly (we need $\varepsilon_1 \leq \exp(-200p^2/n)$). However, we can do a simple error reduction as in Lemma~\ref{lem:repetition} to ensure that Lemma~\ref{lem:main} is applicable. More precisely, choose $r$ to be a large enough absolute constant so that $\varepsilon_1^{r} \leq \frac{1}{2}\exp(-200p^2/n).$ Now, by Lemma~\ref{lem:repetition} there is a probabilistic polynomial $\bm{P^{(r)}}$ of degree at most $r\cdot \deg(P)$ such that 
	\begin{align*}
	\prob{\bm{v}\sim \{0,1\}^n_{n/2},\bm{P^{(r)}}}{\bm{P^{(r)}}(\bm{v}) = 0} &\leq \varepsilon_1^{2r} \leq \frac{1}{2}\exp(-200p^2/n), \text{ and }\\
	\prob{\bm{v}\sim \{0,1\}^n_{n/2-p},\bm{P^{(r)}}}{\bm{P^{(r)}}(\bm{v}) \neq 0} &\geq (1-\varepsilon)^{r} \geq 1-r\varepsilon \geq \frac{9}{10}
	\end{align*}	 
	where for the last inequality we used the fact that $\varepsilon$ is smaller than some absolute constant. 
	
	By a simple union bound, there is a fixed polynomial $P'\in \F_p[x_1,\ldots,x_n]$ of degree $r\cdot \deg(P) = O(m)$ such that
	\begin{align*}
	\prob{\bm{v}\sim \{0,1\}^n_{n/2}}{P'(\bm{v}) = 0} & \leq \exp(-200p^2/n), \text{ and }\\
	\prob{\bm{v}\sim \{0,1\}^n_{n/2-p}}{P'(\bm{v}) \neq 0} & \geq \frac{1}{2}
	\end{align*}	 	
	Hence, applying Lemma~\ref{lem:main} to the polynomial $P'$, we get $\deg(P') = \Omega(p) = \Omega(\sqrt{n\log(1/\varepsilon)}).$ This yields the desired lower bound on $m.$
	\end{proof}

	\subparagraph*{Acknowledgements.} The author is grateful to Mrinal Kumar, Nutan Limaye, Utkarsh Tripathi and S. Venkitesh for useful discussions, feedback, and encouragement. The author thanks Nutan Limaye for suggesting the robust version of Galvin's problem as an application.  The author is also grateful to the anonymous referees of STOC 2020 and TheoretiCS for their corrections and suggestions. In particular, a referee for the TheoretiCS submission pointed out an extension to the main lemma (Lemma~\ref{lem:main-extension}).

	\printbibliography
	
	\appendix
	
	\section{Lemma~\ref{lem:Hegedus} is implied by Lemma~\ref{lem:main} (up to constant factors)}
	\label{appsec:lem2vslem1}
	
	The following claim shows that if there is a $P$ satisfying the hypotheses of Lemma~\ref{lem:Hegedus}, then there is also a polynomial $Q$ of degree at most $\deg(P)$ satisfying a stronger property, namely, that of not vanishing at too many points of $\{0,1\}^n_{k+q}.$
	\begin{claim}
	\label{clm:remark1}
	Let $\mathbb{F}$ be a field of characteristic $p > 0$. Fix any positive integers $n,k,q$ such that $k\in [q,n-q],$ and $q$ a power of $p$. If there is a polynomial $P\in \F[x_1,\ldots,x_n]$ is any  polynomial that vanishes at \emph{all} $a\in \{0,1\}^n_k$ but does not vanish at \emph{some} $b\in \{0,1\}^n_{k+q}$, then  there is a $Q\in \F[x_1,\ldots,x_n]$ of degree at most $\deg(P)$ such that $Q$ vanishes at all $a\in \{0,1\}^n_k$ but is non-zero at at least a $(1-1/p)$ fraction of the points in $\{0,1\}^n_{k+q}.$
	\end{claim}
	
	\begin{proof}
	Let $d = \deg(P)$. Assume without loss of generality that $P(b) = 1.$ Note that $P$ is the solution to the system of linear equations defined by the following constraints on polynomials of degree at most $d$.
	\begin{align*}
		|a| = k &\Rightarrow P(a) = 0 \\
		P(b) &= 1.
	\end{align*}
	As the above linear system is over $\mathbb{F}_p\subseteq \mathbb{F},$ we note that we may assume that $P\in \mathbb{F}_p[x_1,\ldots,x_n].$ From now on, we assume that $\mathbb{F} = \mathbb{F}_p.$

	Consider the degree-$d$ closure $C = \mathrm{cl}_d(\{0,1\}^n_k).$ By the existence of $P$, we see that $b\not\in C.$ However, by symmetry, this implies that no point $b'\in \{0,1\}^n_{k+q}$ lies in $C$. 
	
	Let $V_{d,k}$  denote the vector space of all multilinear polynomials of degree at most $d$ that vanish at all points in $\{0,1\}^n_k.$ Let $\bm{Q}$ be a uniformly random element of $V_{d,k}$. For any $c\in\{0,1\}^n \setminus C$, standard linear algebra implies that $\bm{Q}(c)$ is a uniformly random element of $\mathbb{F} = \mathbb{F}_p.$ In particular, for any $b'\in \{0,1\}^n_{k+q}$, we see that
	\[
	\prob{\bm{Q}}{\bm{Q}(b')\neq 0} = 1-1/p.
	\]
	In particular, there is a $Q\in V_{d,k}$ that is non-zero at at least a $(1-1/p)$ fraction of points in $\{0,1\}^n_{k+q}.$ This yields the statement of the claim.	
	\end{proof}
	
	\section{Proof of Lemma~\ref{lem:string} (the string lemma)}
	\label{appsec:string-lemma}
	
	We begin by recalling the statement of the lemma.

        \lemstring*
	
	\begin{proof}
	 Assume that $|u| = \ell, |v| = m$ and $|w| = \ell+m = n.$  We will show in fact that both $u$ and~$v$ are powers of the same non-empty string $z$. This will clearly imply the lemma.
	
	The proof is by induction on the length of $w$. The base case of the induction corresponds to $n=2$, which is obvious.
	
	We now proceed with the inductive case. Assume w.l.o.g. that $\ell \leq m.$ As $uv = vu$, we see that the first $\ell$ symbols in $v$ match those of $u$, and hence we have $v = uv'$ for some $v'\in \{0,1\}^{m-\ell}.$ If $\ell = m,$ this implies that $u=v$ and we are immediately done. Otherwise, we see that $w = uv'u = v'uu$ for a non-empty string $v'$. Hence, we have $uv' = v'u$. By the induction hypothesis, we know that both $u$ and $v'$ are powers of some non-empty $z$. Hence, so is $v$. This concludes the proof.
	\end{proof}
	
	\section{Proof of Claim~\ref{clm:hypergeometric}}
	\label{appsec:hypergeometric}
We first restate the claim.
        \clmhypergeometric*
	
	\begin{proof}
	It suffices to show that for each $k\in \{0,\ldots,\lfloor m/2\rfloor-1\}$,
	\begin{equation}
	\label{eq:hypergeometric}
	\frac{\binom{n/2}{\lfloor m/2\rfloor - k-1}\binom{n/2}{\lceil m/2\rceil + k+1}}{\binom{n/2}{\lfloor m/2\rfloor-k}\binom{n/2}{\lceil m/2\rceil+k}}\leq \exp(-\Omega(k/m)).
	\end{equation}
	The claim then follows by a simple induction on $k-\ell$.
	
	To prove (\ref{eq:hypergeometric}), we proceed as follows. By an expansion of binomial coefficients in terms of factorials, we see that 
	\begin{align*}
	\frac{\binom{n/2}{\lfloor m/2\rfloor - k-1}\binom{n/2}{\lceil m/2\rceil + k+1}}{\binom{n/2}{\lfloor m/2\rfloor-k}\binom{n/2}{\lceil m/2\rceil+k}} 
		&= \frac{(\lfloor m/2\rfloor -k)(n/2-(\lceil m/2\rceil+k))}{(n/2-(\lfloor m/2\rfloor - k-1))(\lceil m/2\rceil + k+1)}\\
		&\leq \frac{\lfloor m/2\rfloor -k}{\lceil m/2\rceil + k+1}\\
		&\leq \frac{(m/2)-k}{(m/2)} \leq 1-2k/m \leq \exp(-2k/m). \qedhere
	\end{align*}
	\end{proof}

\end{document}